\documentclass[12pt]{article} % Equation numbers right
\usepackage{theorem}\usepackage{harvard}\usepackage{amsmath,amssymb, bbm}
\usepackage{amsfonts,graphicx,verbatim,psfrag}\usepackage{graphics}
\usepackage{verbatim, hyperref}\usepackage{rotating}\usepackage{epsfig}\usepackage{color}
\usepackage{pgf,pgfarrows,pgfnodes,pgfautomata,pgfheaps,pgfshade}

\makeatletter\def\@biblabel#1{\hspace*{-\labelsep}}\makeatother
\newcommand{\be}{\begin{equation}} \newcommand{\ee}{\end{equation}}
\newcommand{\bea}{\begin{eqnarray}} \newcommand{\eea}{\end{eqnarray}}
\newcommand{\beas}{\begin{eqnarray*}} \newcommand{\eeas}{\end{eqnarray*}}
\newcommand{\bit}{\begin{itemize}}
\newcommand{\eit}{\end{itemize}}
\newcommand{\ben}{\begin{enumerate}}
\newcommand{\een}{\end{enumerate}}

\newcommand{\ba}{\begin{array}} \newcommand{\ea}{\end{array}}

\newcommand{\lbar}[1]{\mbox{\b{$#1$}}}

\newcommand{\mb}{\mathbb}
\newcommand{\ep}{\emptyset}
\newcommand{\mc}{\mathcal}

\newcommand{\noi}{\noindent}

\newcommand{\ve}{\varepsilon}

\newcommand{\Bthe}{\begin{theorem}}
\newcommand{\Ethe}{\end{theorem}}
\newcommand{\vphi}{\varphi}

\newcommand{\ind}{\mathbbm{1}}

{\theoremstyle{plain}} \theoremheaderfont{\scshape}

\newtheorem{theorem}{Theorem}

\newtheorem{proposition}{Proposition}
\newtheorem{lemma}{Lemma}

\newtheorem{assumption}{Assumption}
\newtheorem{definition}{Definition}
\newtheorem{observation}{Observation}

\newenvironment{proof}{\noi {\it Proof.} }

\begin{document}

\setlength{\parindent}{12pt}
\setlength{\parskip}{.3cm}

\title{\vspace{-2cm}\LARGE{Can Society Function Without Ethical Agents?\\ An Informational Perspective}\author{Bruno Strulovici\\ Northwestern University \\ }\thanks{This project has benefited from numerous conversations and comments, particularly from Alex Frankel, Simone Galperti, Meg Meyer, Francisco Poggi, Doron Ravid, Larry Samuelson, Ron Siegel, Ludvig Sinander, Boli Xu, and Richard Zeckhauser. Early versions of this project were presented at Microsoft Research New England (Summer 2015), Boston College, Washington University, Universitat Pompeu Fabra, Universitat Aut\`{o}noma de Barcelona, Johns Hopkins University, the University of Chicago, Cambridge University INET Economic Theory Conference (2016), Penn State University, Cornell University, Arizona State University, the North American Meeting of the Econometric Society (2017), Stanford University, Duke University, Northwestern University, Bocconi's workshop on Advances in Information Economics, Kyoto University, Tokyo University (summer school), Stanford SITE 2017, the Transatlantic Theory Workshop 2017, Universidad Carlos III Madrid,  Columbia's Economic Theory Conference (2017), Yale University, the Toulouse School of Economics, LACEA/LAMES 2018, the Summer school of the Econometric Society (Sapporo, 2019), and the RI/PE conference at UC San Diego. Part of this research was developed while I was visiting Microsoft Research New England (2015) and Harvard University (2016) whose hospitalities are gratefully acknowledged.}}
\date{\small{\today}}
\maketitle
%\vspace{-1cm}
\begin{abstract} Many facts are learned through the intermediation of individuals with special access to information, such as law enforcement officers, officials with a security clearance, or experts with specific knowledge. This paper considers whether societies can learn about such facts when information is cheap to manipulate, produced sequentially, and these individuals are devoid of ethical motive. The answer depends on an ``information attrition'' condition pertaining to the amount of evidence available which distinguishes, for example, between reproducible scientific evidence and the evidence generated in a crime.  Applications to institution enforcement, social cohesion, scientific progress, and historical revisionism are discussed.\end{abstract}
%$\quad$ \textbf{Keywords:} Social learning, Informational cascades, Institution Design, Fake News\\
%$\quad$ \textbf{JEL Codes:} [\textit{add JEL codes}]

\newpage

\section{Introduction}

\subsection{Mediated learning}

Many facts must be learned through agents with a specific expertise or access to information. For example, the net benefits of a vaccine, the validity of a mathematical theorem, the correct resolution of a crime, the historiography of an event, and the anthropogenic nature of climate change cannot be directly verified or disproved by the average citizen, and there is no ``public epiphany'' at which the truth is exogenously revealed to all. In these and many others instances, citizens have no choice but to rely on intermediaries to learn anything about the fact of interest, a situation that we will call {\em mediated learning}.

To succeed, mediated learning relies on the investigative efforts and truthfulness of agents who can make false, misleading, or uninformed statements and are subject to biases, pressure, ambition, and other considerations that may distort their behavior. In philosophy and various social sciences, it is not uncommon to consider agents whose rule of behavior is to act truthfully.\footnote{One branch of epistemic philosophy concerns the vulnerability of testimony, i.e., the fact that a speaker can lie. This vulnerability is resolved through behavioral assumptions driven by norms of truthfulness and principles of ethical behavior, such as Grice's ``cooperative principle'' (H.P Grice (1975)). Computer scientists consider the related problem of communication through potentially adversarial intermediaries, known in the literature as the {\em Byzantine generals} problem (Lamport, Shostak, and Pease (1982)). This literature assumes the existence of ``loyal'' generals, who obediently follow the communication protocol set out by the planner, like Grice's cooperative speaker. In sociology, it is common to assume the existence of pro-social norms that coexist with and sometimes subsume individual incentives, and facilitate truthful behavior. See Granovetter (2017) for a recent comparison of the paradigms in economics and sociology.} By contrast, economists usually conceptualize truthful behavior as the result of properly calibrated incentives, given to agents who are devoid of ethical motives.

This paper examines the feasibility of mediated learning through agents who are {\em non-ethical} in the sense that their preferences do not directly depend on the fact to be learned. For example, a prosecutor who cares only about convicting a defendant, regardless of his guilt, is non-ethical according to this definition, while a prosecutor who cares about a defendant's actual guilt may be ethical.\footnote{Caring about the fact to be learned is only a necessary condition for being ``ethical.'' For instance, a prosecutor who wants to convict innocent defendants and acquit guilty ones can obviously not be called ethical. This observation has no bearing on the formal analysis of this paper, which focuses on truth-independent preferences. The case of mediated learning with ethical agents is discussed in Sections~\ref{ssec-ethical-necessity} and~\ref{sec-discussion}.}

%\footnote{The latter behavior is consistent with Landes' (1971) analysis of the courts, which models prosecutors as maximizing the expected number of convictions.}\\

\subsection{Information Attrition}

The cornerstone of the analysis is the concept of {\em information attrition}, which captures the idea that information about a given fact may be in limited supply. Consider the problem of determining whether a given athlete used a banned substance during a given competition. Controls typically consist of two blood (or urine) samples, ``A'' and ``B''. If the first blood sample, ``A'', tests positive (suggestive of doping), the agency storing the second sample, ``B'', is then asked to test the second sample. At this point, the second sample is the only direct source of evidence left about whether the athlete ingested a banned substance at the event: the amount of evidence concerning this fact has shrunk, and each blood test reduces the amount of information available in the case.

Information attrition may be caused by various factors. The first one, outlined above, is that the process of learning about a fact sometimes requires transforming the evidence in a way that prevents its use by subsequent investigators: the signals are {\em disposable}.\footnote{Other examples of disposable signals include human subjects in experimental psychology: once a person has been exposed to a particular experiment, this person is irremediably affected and no longer a good subject for the same experiment. In particle physics and quantum mechanics, signal acquisition is subject to the {\em observer effect}, whereby the observation of a phenomenon necessarily affects this phenomenon.} Second, information attrition may be due to the exogenous degradation of evidence: as time goes by, evidence may deteriorate and become uninformative, for example due to poor storage conditions. Third, information attrition may be caused purposefully by individuals who tamper with or destroy some evidence. Finally, information attrition may stem from social considerations. For example, if some researchers claim to have found that a particular drug or treatment is harmful to human subjects or animals, it may become politically infeasible to run more experiments.

\subsection{Core Argument}

To appreciate how information attrition affects mediated learning, let us reconsider the blood-testing example and, specifically, the agency in charge of testing the second blood sample. Since this agency possesses the only sample left, it can lie at no risk of being contradicted. If, for instance, the agency stands to gain publicity from incriminating the athlete, it can do so with impunity. If instead the agency benefits from exculpating the athlete (perhaps due to the pressure exerted by some sport governing instance), it can also do this at no cost. And if the agency is indifferent between incriminating the athlete and absolving him, it has no incentive to incur the cost of running the test to begin with. In all these cases, the second agency's report is untethered to the truth.

Consider now the agency in charge of the first blood sample. Whatever report this agency produces, the second agency's report will not be based on the truth, as shown in the previous paragraph. The first agency's expected payoff is therefore independent of the content of its blood sample. Like the second agency, it has no incentive to report the true content of its blood sample, and incentives unravel. In this example, mediated learning is feasible only if the agencies in charge of the tests care directly about the truth, i.e., only if they have ethical preferences.\footnote{In practice, the same laboratory is often in charge of storing and testing {\em both} samples, which may lead to even more direct incentives problems, as in the case of nationally organized doping schemes.} When agencies are unethical, mediated learning is infeasible regardless of the agencies' material incentives.\footnote{If there are three or more agencies, the results are the same when the agencies proceed sequentially, regardless of the rule (e.g., majority rule) used to convict the athlete or incentives given to the agencies. Section~\ref{ssec-alternative} discusses the case in which agencies move simultaneously as well as other monitoring structures.}

\subsection{Reproducible Evidence and Incentive Design}

Some investigations are immune to information attrition. Consider the question of determining whether a mathematical proof is correct. The proof remains available for anyone to read no matter how many times it has been checked in the past. The unraveling argument of the previous section no longer applies and, as this paper shows (Theorem~\ref{the-pos-statement}), mediated learning by unethical agents is feasible provided that agents' material incentives are designed appropriately. The incentives that deliver successful mediated learning have an intuitive structure: they reward an agent whose report is vindicated by subsequent findings and punish him otherwise, and the magnitude of the rewards and punishments of each agent depends on how surprising the agent's report is relative to earlier beliefs about the state.

Information attrition does not arise, either, in scientific inquiries that rely on reproducible experiments. For instance, no matter how many times physicists measure the weight of an electron, the experiment can be replicated more times, ad infinitum. Knowledge based on reproducible evidence can also, given the proper incentives, be protected from learning failures caused by unethical agents.\footnote{The theory can be used to predict a correlation between the establishment of truthfulness oaths and contexts where information suffers from attrition. For instance, it  may explain why formal oaths are present in court, but not in physical sciences. See Section~\ref{ssec-oaths}.}

\subsection{Analytical Challenge: Uncertain and Endogenous Attrition}

To assess the scope of the argument illustrated by the blood-testing example, this paper considers an analytical framework based on sequential learning through public reports that generalizes the argument in three important ways. First, the supply of evidence available in a case may be unknown a priori, even to investigators. Second, investigators may not know how past investigators have affected the evidence available in the case. Third, the number of potential investigators may be large and a priori unknown. For example, suppose that a criminal has left behind ten pieces of forensic evidence, which may be discovered through some costly investigative work. Investigators may not know a priori the exact number of pieces left behind. Moreover, if an investigator inherits the case from a previous investigator, he may not know how diligent the previous investigator has been and, hence, what fraction of the evidence has already been discovered, fabricated, or destroyed. An investigator thus faces exogenous and endogenous uncertainty concerning the amount of discoverable evidence.

This paper provides conditions on the probability distribution of the supply evidence, taking into account the impact of each agent on the evidence, under which mediated learning is feasible, and under which it is not (Theorems~\ref{the-main-result} and~\ref{the-pos-statement}). When the necessary conditions are violated, mediated learning by unethical agents is impossible in a strong sense: even when i) there is an unbounded sequence of potential investigators, and ii) investigators' incentives may be arbitrarily designed and administered without any commitment or agency problem, there does not exist any equilibrium in which at least one investigator provides an informative report with positive probability.

It is easy to design incentives for which mediated learning always fails or, given some incentive structure, to construct an equilibrium for which mediated learning fails. The main analytical challenge is to show that, in the presence of information attrition, {\em mediated learning fails for all incentive structures and all equilibria}, in the strong sense that {\em all agents fail to reveal any truth about the fact} with probability one.

Formally, this paper studies a sequential game of incomplete information in which the state variable is a probability distribution over the set of evidence that remains to discover. This probability distribution, defined on an infinite-dimensional space, represents agents' belief about the supply of evidence and does not have simple monotonicity properties. For instance, while the discovery of a piece of evidence may suggest that the remaining supply of evidence is now smaller by one piece, this discovery may be interpreted as the tip of an iceberg of evidence, pointing to many more pieces to discover. A discovery could also, if it contradicts past reports, indicate that previous investigators have lied and that whatever evidence they purported to have found (and thus ``removed'' from the supply of evidence) was in fact fabricated, resulting in a more optimistic belief about the amount of evidence that actually remains to discover.

An agent's belief about the evidence supply affects his incentives directly, through the probability that he discovers new evidence, as well as indirectly, through the probability that subsequent agents also look for evidence. An agent cares about later agents' beliefs about the supply of evidence, their beliefs about subsequent agents' beliefs, and so on. Moreover, an agent can manipulate other agents' beliefs by lying in his report.

Information attrition does not rule out per se the possibility that many agents work, or that many pieces evidence remain to be discovered, but it suggests some negative correlation between these events. To rule out the existence of an informative equilibrium, one has to control the sequence of agents' beliefs. This is achieved by showing that the probability that an agent discovers evidence is probabilistically linked to the impact that this agent's report has on subsequent agents' beliefs. This key link is established in  Proposition~\ref{pro-general-mixed}.

\subsection{Virtual Attrition Paradox. Application to Human Subjects}\label{ssec-virtual}

The necessary conditions are violated---and mediated learning is impossible---whenever information consists of disposable signals whose number has a bounded support, no matter how large the bound. For example, suppose that there are initially 1000 signals. Successful learning requires at least one agent to use one of the signals, which can be incentivized only if a subsequent agent uses (with positive probability) a second signal, which can be incentivized only if a third agent uses a third signal with positive probability, and so on. To incentivize mediated learning, agents must, with positive probability, reach a point at which i) the probability that there remain 2 or more signals is arbitrarily low, ii) the probability that some agent makes an informative report is positive. This is essentially the situation faced by the two laboratories in the blood-testing example, and we know that these two conditions are incompatible. Mediated learning must therefore fail down this path, which causes incentives to unravel all the way back to the first agent and leads to a complete and global failure of mediated learning.\footnote{In herding models (Bikhchandani, Hirshleifer, and Welch (1992) and Banerjee (1992)), social learning can fail after so much public information has already been acquired that agents forgo their own information. Here, by contrast, learning fails entirely from the first agent onward.}

This yields the following paradox: when the supply of disposable signals is bounded, mediated learning fails because agents anticipate that signals must eventually become scarce with positive probability if mediated learning is successful. This anticipation destroys agents' incentives to seek signals in the first place, implying that the scarcity never materializes: it remains virtual and the supply of evidence remains intact as mediated learning fails entirely.

To appreciate the logical consequences of this paradox, consider a research question that must be investigated through human subjects, such as a question in experimental psychology. The experiment can be run only once on any given subject: exposure to the experiment affects a subject’s awareness, knowledge, and response to future experiments, as can be most vividly intuited from experiments such as Zimbardo's 1971 Stanford prison experiment (Haney, Banks, and Zimbardo (1972)).

In these cases, the more experiments are performed, the fewer the number of subjects available for further experiments. Even when the pool of potential subjects is a large, virtual attrition limits the scope for experimental replication and, hence, for incentivizing researchers' truthful behavior. Remarkably, Zimbardo later admitted to largely influencing the participants of his prison experiment, which significantly reduces the experiment's scientific value.\footnote{Zimbardo (2007, p. 55) admits to directing the guards to handle the prisoners in the following terms: ``We can create boredom. We can create a sense of frustration. We can create fear in them, to some degree. We can create a notion of the arbitrariness that governs their lives, which are totally controlled by us, by the system, by you, me \ldots They'll have no privacy at all, there will be constant surveillance---nothing they do will go unobserved. They will have no freedom of action. They will be able to do nothing and say nothing that we don't permit. We're going to take away their individuality in various ways \ldots  In general, what all this should create in them is a sense of powerlessness. We have total power in the situation. They have none.''}

\newpage
\section{Implications of the Theory}

\subsection{Political Consequences of Information Attrition}

Consider the perspective of a citizen who is {\em cynical} in the sense that he does not trust information intermediaries to behave ethically. In this citizen's mind, all agents involved in the learning process are purely motivated by unethical motives and are collectively aware of this.

To see how information attrition affects the view of cynical citizens, let us first revisit the blood-testing example. Given that the blood samples are subject to information attrition, it is reasonable---indeed, rational---for a cynical individual to treat as uninformative any finding reported by the blood-testing agencies, for the reasons explained above.\footnote{Blood tests could in principle be certified by a third party, who would check wether the agency did its job properly and thus increase the trust in its report. One would then have to understand the third party's incentives. This ``monitoring the monitor'' structure is discussed in Section~\ref{ssec-alternative}.} As a result, two citizens with cynical beliefs about agencies’ behavior and otherwise different views of the world may rationally entertain completely different beliefs about whether a particular athlete used banned substance, even after agencies have made their reports public: mediated learning fails to convince citizens and bring their views closer.

Mediated-learning failures can have severe consequences for social cohesion and political stability. When a politician is accused of corruption, for instance, the average citizen must rely on declarations made by intermediaries, such as officials and journalists, to learn anything about whether the corruption charge is true or, instead, an attempt to smear and neutralize some politician. There is no {\em deus ex machina} available to lift all confusion and finally reveal the truth to the public. Moreover, information attrition may be a concern in these environments. In a corruption case, for instance, incriminating documents can be destroyed and witnesses intimidated or eliminated.\footnote{The same observations hold with regard to government agencies suspected of abusing their power or violating some rules. Examples of evidence destruction by governmental agencies abound, even in the world's strongest democracies. In the Unites States, for instance, CIA director Richard Helms ordered in 1973 that all documents pertaining to the CIA's infamous MK Ultra program on mind-control experiments be destroyed (``An interview with Richard Helms'', \url{https://www.cia.gov/library/center-for-the-study-of-intelligence/kent-csi/vol44no4/html/v44i4a07p_0021.htm}).
 In 2005, the CIA's Director of Operations at the time ordered the destruction of all interrogation tapes of Abu Zudaydah and Abd al-Rahim al-Nashiri that featured ``enhanced'' interrogations (``Tapes by C.I.A. Lived and Died to Save Image,'' \url{https://www.nytimes.com/2007/12/30/washington/30intel.html}.)}

Now let us consider citizens holding very different priors, perhaps based on their party affiliation or ideology, about whether some politician was guilty of corruption or some governmental agency abused its power. If these citizens are cynical, in the specific sense given above, they can reject official statements and journalistic reports about the case and maintain strong disagreements.

The theory thus provides an explanation for why even reasonable citizens may remain divided on questions that are subject to information attrition.%\footnote{This result is related to the question studied by Cripps et al. (2008) and Acemoglu et al. (2016): do rational agents have convergent beliefs as they receive more signals over time? In Acemoglu et al. convergence may fail because individuals disagree over how to interpret signals. Here, by contrast, individuals agree that signals cannot be trusted.} 
For such questions, the absence of trust in investigative institutions is rational and can perpetuate polarization. As a corollary, this paper also proposes a specific mechanism to explain why eroding citizens' trust in institutions harms social cohesion.

\subsection{Ethical Necessity}\label{ssec-ethical-necessity}

The arguments presented so far bring us to one of the paper's key motivating questions: is ethical behavior necessary for society to function? {\em Ethical necessity} is not salient in economic analysis, which typically focuses on ``selfish'' agents evolving within the boundary of well-defined institutions, such as markets or democratic institutions, whose enforcement is taken for granted.\footnote{Economists have studied various forms of non-selfish behaviors, such as pro-social behavior arising in dictator and ultimatum games and various forms of altruism. Unlike earlier work, this paper does not study ethical behavior per se, but the necessity of ethical behavior for society to function. Self-interest remains the default assumption in most economic models and is deeply rooted in the discipline. For example, Edgeworth (1881) observed that ``self-interest is the first principle of pure economics.''} In reality, institutions cannot be taken for granted: if agents are unethical, they may attempt to violate these institutions in various ways, and it is unclear why agents should be presumed to act selfishly within the behavioral boundaries imposed by these institutions but ethically respect these boundaries. Such an assumption amounts to a ``heroic dichotomy'' that, at the very least, deserves closer inspection.

By focusing on mediated learning, this paper provides a tractable framework to study ethical necessity.\footnote{The present framework relies on a sequential learning structure. One could consider alternative structures. For example, a central agency may ask several intermediaries to seek and report the truth simultaneously and independently of one another. Under this parallel structure, the incentives of such a centralized agency must however also be scrutinized, which reintroduces a sequential element. This and other designs are discussed in Section~\ref{sec-discussion}.} Since mediated learning is involved in the resolution of criminal cases, political corruption, and the determination of other possible violation of institutions, successful mediated learning is a necessary condition for society to function: whenever mediated learning requires ethical behavior, so does society.

While ethical necessity has not been a salient issue in economics, it did receive some attention from economists.\footnote{For instance, Myerson (2006) shows that in a federalist democracy, the existence of some virtuous politicians can guarantee that democracy succeeds either at a national or a provincial level, whereas democracy can fail at both levels when such politicians are absent. By contrast, Glazer and Rubinstein (1998) describe  an information aggregation environment in which if all agents are purely concerned with achieving the social optimum, there always exist equilibria in which the optimum is not achieved whereas if all agents are {\em also} (but not only) concerned with their individual recommendation being followed, this uniquely selects the socially optimal equilibrium.} Hurwicz (2007) explicitly discusses the existence of ``intervenors,'' which he defines as ethical monitors in a monitoring-the-monitor problem, and expresses his personal belief in the existence of intervenors. Unlike the present paper, however, Hurwicz argues that intervenors are not needed for successful monitoring hierarchies.\footnote{Rahman (2012) proposes an approach that is well suited for repeated monitoring tasks, such as controls used for airport security or in sting operations: a principal can ask agents to violate the rules on purpose to check whether these violations are caught by the monitor. Such violations are detected ``for free'' by the principal since he instigates them. The approach is well suited when violations can be faked at little social cost, the principal has commitment power, and collusion between the principal and agents is impossible.} Hurwicz describes an environment with three agents $A,B,C$, in which $B$ monitors $A$'s actions, $C$ monitors $B$'s monitoring of $A$, $A$ monitors $C$'s monitoring of $B$'s monitoring $A$, and so on. This ``Hurwicz triangle,'' in which the infinite monitoring hierarchy is folded in loops going through the three agents, which Hurwicz does not formally analyze, omits the possibility of corruption across monitors, studied in Strulovici (2019).\footnote{Levine and Modica (2016) consider a similar structure, in which agents in a group take some initial action, then verify the action taken by their neighbor, then verify the earlier verification task of their neighbors, and so on.}

Holmstr\"{o}m (1982) shares a similar concern for his analysis of moral hazard in teams, noting in his conclusion that ``another important issue relates to monitoring hierarchies. (...) The question is what determines the choice of monitors; and how should output be shared so as to provide all members of the organization (including monitors) with the best incentives to perform?''

Ethical necessity is implicit in Becker and Stigler's (1974) study of wrongdoing and malfeasance by enforcement officers. The authors assume the existence of exogenous signals about the behavior of enforcement officers. The existence of a reliable second level of monitoring is thus taken for granted, which raises the question of how these signals are generated.\footnote{Milgrom, North, and Weingast (1990) explicitly study the enforcement of market institutions. They consider the role of law merchants to discipline trade in medieval Europe, and they do consider some of the law merchants' incentives to lie and take bribes.}

\subsection{Historical Revisionism}

Historical revisionism is another instance of mediated learning, in which information attrition plays an important role. Understanding historical events is of obvious importance, not only to learn lessons from the past but also to assess claims based on such events, such as territoriality or reparation claims. %\footnote{Social construct, ideology, international relations. Cite.}
Mediated learning is necessary because citizens cannot directly verify historical events.\footnote{This point is obvious with regard to events that took place before citizens' lifetime. Even with regard to contemporaneous events, aggregating information and forming a global picture of an event is a highly complex task that requires expertise, time, and a special access to information. The pitfalls of such information aggregation have been famously illustrated by Stendhal's  {\em La Chartreuse de Parme} whose protagonist, Fabrice del Dongo, takes part to the Battle of Waterloo with the Napoleonic army but construes a completely erroneous version of the battle.} They must rely on experts and officials to access and correctly interpret archives, artifacts, and other sources of information. Information attrition is both exogenous (e.g., witnesses die) and endogenous (e.g., documents may be destroyed on purpose).

To give a concrete example,\footnote{Huq (2018) discusses similar, very recent examples, in which the possibly false threats of terrorism or coups were used to weaken democratic institutions and shift power to more authoritarian regimes.} consider the fire of the German parliamentary building (Reichstag) on February 27, 1933. The importance of this event can hardly be overstated. The Nazis, who had lost seats in the previous parliamentary election, claimed that the fire had been caused by communists and used this event to pressure president Hindenburg into imposing martial law on Germany (the ``Reichstag fire decree'') and arrest and weaken communists. This allowed the Nazi's to form a majority coalition after the March 5, 1933 parliamentary elections, consolidated Hitler’s power, and led to the Enabling Act.

While communist involvement in the fire has long been rejected, there was until recently a consensus among mainstream historians that the Reichstag fire had {\em not} been caused by the Nazis. Fritz Tobias, one of the most respected historians on this subject in the postwar period, published a series of articles purporting to show that van der Lubbe, the person who was convicted for the arson, had acted alone.\footnote{The articles were published in {\em Der Spiegel} under the title ,\hspace{-.03cm},Stehen Sie auf, van der Lubbe!,'' in 1959 and 1960, in issues 43 to 52.} Historians accepted Tobias' version of the event until 2001, when two historians studying Gestapo archives came to the conclusion that it was a group of SA officers who had set the fire (Bahar and Kugel (2001)). Hett (2014) used recent scientific advances to convincingly argue that it would have been impossible for a single individual to set the fire. In 2019, an affidavit written in 1955 by former SA Hans-Martin Lennings discovered in Tobias' files was published by RedaktionsNetzwerk Deutschland, which stated that Lennings and other SAs had driven van der Lubbe from an infirmary to the Reichstag when the fire had already started, effectively setting up van der Lubbe.

Information attrition took several forms: all but one of the SA officers who were allegedly involved in the Reichstag fire were killed (and, hence, silenced) during the Night of the Long Knives; van der Lubbe was beheaded in 1934 for his alleged role in the arson; and any forensic evidence about the Reichstag fire has been long gone. Furthermore, Fritz Tobias hid Lennings' affidavit, which contradicted Tobias' single perpetrator theory, until his death.\footnote{Tobias is suspected of having tried to protect former Nazi officers after the war and had some incentive to dissimulate the Nazis' role in the Reichstag fire.} To this day, the strongest case for Nazi involvement thus seems to come from Lennings' affidavit, and hence boils down to one man's statement. It is unclear what Lennings' motivation for involving the SA may have been, except perhaps for setting the record straight. For observers who doubt Lennings' statement, the question of who set the Reichstag on fire may reasonably represent a failure of mediated learning. Lutjens (2016) believes that ``the continuous
reshaping of the Reichstag fire by those with a stake in the matter has fragmented the truth beyond recovery.''

\subsection{Existence and Observability of Ethical Agents}

The arguments developed so far assume that agents are known to be unethical. Without this knowledge, mediated learning can be achieved if agents put sufficiently high weight on the probability that other agents are ethical. To see this, let us consider once more the blood-testing example. If the first agency and citizens all believe that the second agency behaves ethically, the first agency can be incentivized to test and truthfully report the content of the first blood sample because its finding can be compared to the second agency's finding, which all believe to be truthful, and the first agency can be held accountable for any discrepancy.

More generally, any agent can be incentivized to behave truthfully as long as the agent believes that subsequent agents are likely to provide informative reports.

Conversely, if some agent is ethical, but other agents believe that he is not,  mediated learning can fail just as when all agents are unethical, for two reasons: First, the findings of the ethical agent are (perhaps, wrongly) believed to be uninformative. Second, precisely because the agent's findings are believed to be uninformative, the findings cannot be used to incentivize other agents to behave truthfully. Therefore, mediated learning requires that agents believe that other agents behave ethically with high enough probability.

The theory thus provides a specific mechanism for why eroding trust in institutions is damaging: society may need everyone to believe in the existence of ethical agents in order to sustain ethical behavior. Events that erode the strength of this belief can have severe consequences for the feasibility of mediated learning and the functioning of society.\footnote{Belief in ethical behavior gets rid of virtual attrition, but not of real attrition. For example, if the second blood-testing agency is, in fact, unethical, it cannot be incentivized to tell the truth: real attrition interrupts mediated learning. But if everyone erroneously believes that the second agency is ethical, then the first agency may be incentivized to behave truthfully.} It may be empirically difficult to distinguish between agents who have ethical preferences and agents who merely behave like ethical agents. The belief in ethical behavior can be self-fulfilling.

\subsection{Organization of Remaining Sections}

Section~\ref{sec-formal} describes the formal model and results. In the baseline model, all agents must incur a cost (which may be arbitrarily small, but strictly positive) for acquiring information. The model is then is then extended to allow for the existence of witnesses, who receive information for free about the state of the world and are subject to idiosyncratic, private shocks that affect their reporting preferences. Section~\ref{ssec-hard} describes the relation between the concept of hard evidence, information attrition, and intermediation. Section~\ref{ssec-alternative} discusses alternative investigation designs and settings. Section~\ref{ssec-oaths} explores how to foster truth-dependent preferences among investigators. All results are proved in the Appendix.

\newpage
\section{Formal Analysis}\label{sec-formal}

\subsection{Baseline Model}

The fact of interest, $\omega\in \Omega$, must be inferred from a sequence $S = (s^1,\ldots, s^{\tilde K})$ of signals, each of which takes values in some finite signal space~$\Sigma$. The sequence $S$ and its length $\tilde K\leq \infty$ are stochastic.%The object of mediated learning is the sequence $S$ itself. There are various ways of describing how information about $S$ maps into information of $\omega$ and will not model them.

In each round $i\geq 1$, a new agent arrives and makes two decisions: First, the agent privately chooses between seeking a signal (``working'') at cost $c > 0$ and doing nothing (``shirking''). Second, the agent publicly sends a message $m_i$ from some finite message space $M$. The agent can randomize his decisions.

Let $S_i$ denote the sequence of signals that remain to discover at the beginning of round~$i$ (in particular, $S_1 = S$). If the agent in round $i$ (hereafter, ``agent~$i$'' or simply ``$i$'') works and $S_i\neq \ep$, then $i$ discovers some element of $S_i$ with probability $\lambda\in (0,1]$. The discovered signal is denoted $s_i$.\footnote{The subscript $i$ denotes the signal discovered in round i (if any), which should not be confused with the ordering of signals in the sequence $S$, which is indexed by a superscript.}{\small{'}}\footnote{No constraint is imposed on how likely each element of $S_i$ is of being discovered by agent $i$. This likelihood could depend arbitrarily on $i$'s identity and on $S_i$, and the discovery of some specific signal may contain information about other signals in $S_i$ that depends on $i$'s identity. Moreover, Theorem~\ref{the-main-result} can be generalized to the case in which the probability $\lambda$ of discovering a signal is nondecreasing in the length of $S_i$.} With probability $1-\lambda$, $i$~discovers no signal. If $S_i$ is empty, $i$ surely discovers no signal.

For simplicity, we will not model the possibility that agents destroy signals without discovering them, or that signals disintegrate exogenously. These additional forms of information attrition would only strengthen the paper's impossibility results (Theorem~\ref{the-main-result} and Theorem~\ref{the-witness}).\footnote{The framework of the possibility result (Theorem~\ref{the-pos-statement}), in which the number of signals is geometrically distributed, can be interpreted as being the result of an exogenous decay of evidence.} The extension is more complex when some agents, such as witnesses, can discover signals for free. This extension is analyzed explicitly in Section~\ref{ssec-witnesses}.

After the information-seeking stage, $i$ sends a report $m_i$ whose distribution in $\Delta(M)$ can arbitrarily depend on what (if anything) $i$ has observed in the information-seeking stage and on the reports $m_1^{i-1} = (m_1,\ldots, m_{i-1})$ made by previous agents.

Entering round $i+1$, we have $S_{i+1} = S_i$ if $i$ did not discover any signal. If $i$ discovered a signal, then $S_{i+1}$ is a subsequence of $S_i$ with length $|S_{i+1}| = |S_i|-1$.

Let $m = (m_1, m_2, \ldots)$ denote the sequence of reports made by all agents. The realized utility of agent $i$ is given by
\be\label{eq-utility}U_i = V_i(m,\omega) - c \ind_{\mbox{$i$ works}}\ee
where $V_i$ takes values in some compact interval $[-P,R]$.

We will say that agent $i$ is {\em non-ethical} if the function $V_i$ is independent of $\omega$, which captures the idea that $i$ does not care directly about the truth. In this case, $V_i$ may be defined on the restricted domain $M^{\mb N}$. $V_i$ can depend arbitrarily on the entire sequence $m = (m_1, m_2, \ldots)$. In particular, the impossibility results presented in this paper hold regardless of whether $V_i$ is an exogenously given utility function or one that is specifically designed (or, at least, influenced) by a regulator or social planner.

To illustrate the various forms that $V_i$ may take, note that $i$ could be punished if his report is contradicted by subsequent investigators, and/or rewarded if his report differs from past investigators' (as in a journalistic scoop). $V_i$ may aggregate a discounted stream of rewards and punishments.\footnote{For example, if $i$ receives utility $v_{i,j}(m_1,\ldots, m_j)$ in round $j\geq i$ and discounts future utility with some factor $\delta <1$, then
\[V_i(m) = \sum_{j\geq i} \delta^{j-i} v_{i,j}(m_1,\ldots m_j).\]}
The formulation captures situations in which $i$'s utility is affected by reports indirectly through the actions that these reports trigger. For example, $i$ could be a prosecutor at a trial, whose outcome $a(m)\in \{\mbox{`guilty', `not guilty'}\}$ depends on the statements of all agents involved in the case. If $i$ is non-ethical, his utility may be modeled by $V_i(m) = R \times \ind \{\mbox{$a(m) =$ `guilty'}\}$, as in Landes (1971). In general, $i$'s realized utility could depend stochastically on other agents' reports. For example, suppose that the number of investigators is stochastic. We can model this by interrupting mediated learning at some stopping time $\tau$, in which case $i$'s utility depends only on $(m_1,\ldots, m_\tau)$. This and similar variations, such as randomizing the order of investigators or how investigator $j$'s report affects $i$'s utility, is  easily encompassed by the model.

Agents have a common prior about the distribution of $S$. The equilibrium concept is (weak) Perfect Bayesian Equilibrium.\footnote{The impossibility results in this paper hold for stronger concepts of equilibrium, such as sequential equilibrium. Reciprocally, the equilibrium constructed to prove the positive result, Theorem~\ref{the-pos-statement}, is a sequential equilibrium.}

\subsection{Main Results}

For each $k\geq 1$, let $F^k = Pr(|S| \geq k)$ denote the prior probability that there are at least $k$ signals to discover at the beginning of the investigation process.

\begin{definition} An equilibrium is {\bf informative} if at least one agent works with positive probability, and {\bf uninformative} otherwise. \end{definition}

Uncovering any modicum of information about $S$ (and, hence, $\omega$) with positive probability, no matter how small, suffices to qualify an equilibrium as ``informative.'' The following theorem shows, however, that informative equilibria fail to exist when information is subject to attrition in a specific sense.

\begin{theorem}\label{the-main-result} Given fixed parameters $(R, P, c, \lambda)$, there exist strictly positive thresholds $\{\lbar F^k \}_{k\geq 1}$ with the following property:
\[\mbox{All equilibria are uninformative if there exists $k\geq 1$ such that  $F^k < \lbar F^k$.}\]
\end{theorem}

For an informative equilibrium to exist, the support of $|S|$ must therefore be unbounded: if not, $F^K = 0$  for some $K$, which violates the threshold conditions regardless of the parameters $(P,R,c,\lambda)$. This only a necessary condition, however: even if $|S|$ has an unbounded support, the thresholds mentioned by Theorem~\ref{the-main-result} imply that the survival function $F^k = Pr(|S| \geq k)$ cannot decrease too fast for mediated learning to be feasible.

When the survival function $F^k$ decreases at most at a geometric rate, there are instances of the model and utility functions $\{V_i\}_{i\in \mb N}$ for which mediated learning is feasible, as indicated the next theorem.

\begin{theorem}\label{the-pos-statement} For any $\rho\in (0,1]$ and $\lambda \in (0,1]$, there exists an instance of the model and utility functions $\{V_i\}_{i\geq 1}$ for which the distribution of $S$ satisfies $F^{k} = \rho^{k-1} F^1$ for all $k\geq 1$ and an informative equilibrium exists.\end{theorem}

To illustrates this positive result, the next section focuses for expositional simplicity on the case of reproducible evidence, which corresponds to $\rho = 1$ (the supply of signals is unlimited), and shows that mediated learning can be made arbitrarily precise as long as the rewards and punishments are high enough.

\subsection{Reproducible Evidence}

Suppose that $S$ consists of infinitely many signals taking binary value, ``$H$'' or ``$L$''.\footnote{The construction used to prove Theorem~\ref{the-pos-statement} is almost identical for $\rho \in (0,1)$ and $\rho = 1$ and the results are qualitatively unchanged. The key is that there remain signals to discover with probability at each stage of the sequence with a probability that is bounded below away from 0. Moreover, if the number of initial signals is generated according to a geometric distribution and each signal decays at some fixed rate, the construction used to prove Theorem~\ref{the-pos-statement} still works because the probability of discovering new signals remains uniformly bounded away from zero.} The signals are conditionally i.i.d.: there is an unknown parameter $\omega \in \{H,L\}$---the underlying fact of interest---such that each signal $\tilde s$ satisfies $Pr(\mbox{$\tilde s$ = ``$H$''}| \omega =H) = Pr(\mbox{$\tilde s$ = ``$L$''}| \omega =L)= \pi \in (1/2,1)$, and the signals are independently distributed conditional on~$\omega$.  Agents' message space is chosen to be binary: $m_i \in \{`\mbox{`$H$''}, \mbox{``$L$''}\} $ for all $i$.

For expositional simplicity, we assume in this section that $\lambda=1$, which means that every agent who works surely discovers a signal.\footnote{The construction used to prove Theorem~\ref{the-pos-statement} is very similar if instead $\lambda <1$.}

In the equilibrium that we consider, it is always in an agent's interest to follow his signal if he acquired one. The relevant pure strategies are: i) work at cost $c>0$ and report one's signal ($m_i = s_i$), and ii) shirk and send a message $m_i \in \{`\mbox{`$H$''}, \mbox{``$L$''}\} $ at no cost.

Let $p_0 = P(\omega=H)$ denote the prior about $\omega$ before the investigation process. For any given equilibrium, let $\gamma_i$ denote the probability that $i$ works and $p_i$ denote the probability that $\omega = H$, both conditional on the past reports $m_1^{i-1}$.

\begin{proposition}\label{pro-positive} For any thresholds $p_-$ and $p_+$ such that  $0 < p_- <  p_0 <  p_+ < 1$, there exist $P,R>0$, utility functions $\{V_i\}_{i\geq 1}$ taking values in $[-P,R]$, and thresholds $\lbar p,\bar p$ such that $0<\lbar p < p_-$ and $1>\bar p> p_+$ for which the following strategy profile constitutes an equilibrium: $\gamma_i = 1$ if $p_i\in (\lbar p, \bar p)$ and $\gamma_i = 0$ otherwise.\end{proposition}

The state $\omega$ can thus be learned with arbitrary precision if the rewards and punishments used to incentivize agents are high enough. In equilibrium, agents work with probability 1 until the posterior belief becomes extreme enough, at which point learning stops.

Since $p_i$ is a martingale and each signal has the same level of informativeness, $p_i$ must exit $[\lbar p, \bar p]$ with probability 1 in the candidate equilibrium. Incentives are provided as follows: if $i$ reported ``$H$'', he gets a reward if $\bar p$ is reached and a punishment if $\lbar p$ is reached, and vice versa. These rewards and punishment depends on the belief $p_i$ before $i$'s report. If $p_i$ was very close to one of the boundaries and $i$'s report takes the posterior away from this boundary, $i$ gets a high reward if the belief process ends up exiting through the other boundary (a low probability event) and a very mild punishment if the belief process ends up crossing the nearby boundary.

\subsection{Witnesses}\label{ssec-witnesses}

Consider now the addition of witnesses who differ from previous agents in two aspects:
\bit
\vspace{-.4cm}\item Witnesses discover a signal for free.
\vspace{-.2cm}\item They are subject to preference shocks that affect which report they prefer to send.
\eit
\vspace{-.2cm}
In each round $i\geq 1$, agent $i$ can be an investigator, identical to the agents of the baseline model, or a witness. Whether $i$ is an investigator or a witness is public information.

If $i$ is an investigator, the structure of $i$'s round, information, and utility is as in the baseline model. If $i$ is a witness, he receives at no cost a signal $s_i\in S_i$ at the beginning of the round, where $S_i$ is the set of signals remaining at the end of round $i-1$. In particular, $i$ can be a witness only if $S_i$ is nonempty. If $i$ is a witness, the sequence $S_{i+1}$ of available signals at the end of round $i$ satisfies $|S_{i+1}| = |S_i|-1$.

At the beginning of round $i$, the probability $\vphi_i$ that $i$ is a witness is zero if $S_i=\ep$ and can take any value in $[0,1]$ otherwise, and depend on past reports $m_1^{i-1}$.\footnote{For example, there could be a fixed subset $\mc N\subset \mb N$ of rounds such that $\vphi_i = \ind_{i\in N}$ (which puts a lower bound on the support of $|S|$). Alternatively, there could be a fixed time $T$ beyond which all witnesses have appeared, so that $\vphi_i \in (0,1)$ for $i\leq T$ and $\vphi_i=0$ for $i> T$. In general, the number of witnesses and the timing of their appearance may be stochastic and depend on past reports.  This captured by allowing $\vphi_i$ to depend arbitrarily on past reports and on calendar time.}

In principle, a witness' signal could be informative about the number of signals that remain to discover and, in particular, about whether information attrition is an issue for subsequent agents.\footnote{For example, $|S|$ could be a finite or infinite with equal probability, and a witness' signal could reveal whether the latter is true, in which case mediated learning may be feasible, as described by Theorem~\ref{the-pos-statement}.} This would significantly weaken the relevance of the initial belief about the supply of information, on which previous theorems are based.

We rule out this possibility and focus on the case in which a witness' signal never increases one's expectations about the total number of signals. This is achieved as follows: the sequence $S$ of signals is obtained by, first, generating an infinite sequence $S^\infty$ of signals, which may exhibit any arbitrary correlation between one another and, second, by truncating this sequence at some integer-valued random variable $\tilde K$ that is independently distributed from $S^\infty$. Agents observe signals in the order of the sequence. Thus, writing  $S= (s^1,\ldots, s^{\tilde K})$ (using superscripts to avoid confusion with the signals discovered in round $i$, which are denoted with subscripts), suppose that $q_i$ signals have been uncovered by the beginning of round $i$, so that $S_i = (s^{q_i+1}, s^{q_i+2},\ldots, s^{\tilde K})$. If $i$ discovers a signal $s_i$, then necessarily $s_i = s^{q_i+1}$ and $S_{i+1} = (s^{q_i+2},\ldots s^{\tilde K})$.

Moreover, we assume that $\tilde K$ has an increasing hazard rate, i.e., $\Pr(\tilde K=k)/\Pr(\tilde K\geq k)$ is increasing in $k$.
\begin{assumption}\label{A-IHR} The total number of signals $\tilde K$ has an increasing hazard rate.\end{assumption}
Intuitively, this assumption guarantees that the more signals have been discovered, the more likely it is that there are no signals left to discover.

After observing his signal $s_i$, witness~$i$ sends a report $m_i\in M$. His realized utility has two parts:
\[U_i(m) = V_i(m) + \epsilon_i(m_i)\]
where $V_i(m)$ plays the same role as investigators' utility function, and $\epsilon_i(m_i)$ is a shock affecting $i$'s preferences.

\begin{assumption} The random variables $\{\epsilon_i(m_i)\}_{m_i\in M_i}$ are independently distributed from one another and from all other variables in the model. They have density functions $\{f_{i,m_i}\}_{m_i\in M_i}$ that are bounded above by some constant $\bar f$.\end{assumption}

The bound $\bar f$ plays a role for witnesses that is similar to the ratio $1/c$ for investigators. Intuitively, the opportunity cost for $i$ to send message $m_i$ instead of $m'_i$ (or vice versa) is of order $1/\bar f$, as explained in Lemma~\ref{lem-order-stat}.

An informative equilibrium is defined as before: there is at least one agent (investigator or witness) who produces an informative message with positive probability. Continuation informative equilibria are defined analogously.

\begin{theorem}\label{the-witness} There exist strictly positive thresholds $\{\lbar F_k\}_{k\geq 1}$ such that if $F^k<\lbar F^k$ for some $k\geq 1$, there does not exist any informative equilibrium.\end{theorem}

\newpage
\section{Discussion}\label{sec-discussion}

Learning through intermediaries fails if the following conditions hold jointly: i) intermediaries do not care about the truth, ii) information is subject to attrition, iii) there is no exogenous, public revelation of the truth at any future time, and iv) intermediaries proceed sequentially.

This result holds for a general class of utility functions, for all equilibria, in the strong sense that nothing at all is learned about the state of the world, and despite the fact that agent incentives can be administered without any further agency problem: whatever rewards and punishments are promised to the agents as function of reporting histories can be perfectly enforced.

This is not to say that mediated learning fails in practice. This paper's impossibility results may be viewed as a reference point: to succeed, mediated learning must break at least one of the four conditions above. In particular, mediated learning can succeed if some intermediaries are motivated by the desire to seek the truth, or by a belief that other intermediaries have such a motivation.

Several ways out are discussed below.

\subsection{Escaping Attrition with Hard evidence}\label{ssec-hard}

In some cases, signals are not disposable. Consider for instance the video footage of a crime, in which the criminal is clearly identifiable. Such a video can be revisited numerous times without being altered, and conceptually resembles the case analyzed earlier of reproducible evidence.

Even this kind of evidence has some limitations. First, the evidence can deteriorate over time. Second, it need not be perfectly reliable. For example, the video footage could be fabricated, as exemplified by the emergence of baffling deepfakes. In this case, the signals that investigators receive, no matter how numerous, are only as reliable as the original source that generates them.

The use of DNA testing also illustrates this issue. The amount of usable DNA samples on a crime scene is finite and thus subject to virtual attrition (Section~\ref{ssec-virtual}). The procedure of DNA testing involves a replication phase, such as PCR amplification or DNA cloning, which creates more ``evidence'' that can be stored and verified by subsequent investigators. However, DNA samples can also be synthesized, i.e., literally ``fabricated,'' to match any desired DNA profile.\footnote{Frumkin, Wasserstrom, Davidson, and Grafit (2010) show the possibility of creating saliva or blood samples with the desired DNA. The authors, as well as subsequent work by other researchers, show that identifying methylation patterns in DNA samples can help distinguish synthetic and natural DNA, although such identification is challenging.} A laboratory can store artificial DNA, or even more simply, it can store DNA samples that were erroneously or malevolently taken outside of the crime scene and presented as coming from the scene.\footnote{A famous example is the ``Phantom of Heilbronn,'' a presumed serial killer whose DNA was found on 40 crime scenes over a fifteen-year span in Germany, Austria, and France. The DNA turned out to belong to a woman working in the factory that made the cotton swabs used to collect DNA samples.} Ultimately, agents who come later in the investigation can either test the DNA material that previous laboratories have left them, which may be the outcome of manipulation, or look for new DNA samples, which brings us back to the problem of information attrition.

The role of technology on mediated learning is complex and deserves a separate exploration. For example, DNA testing, video footage, and other technological advances have increased the set of reliable evidence. However, technology can also be used to manipulate evidence and do so more anonymously than before, increasing the reliance on experts and, hence, on mediated learning.

\subsection{Alternative Designs}\label{ssec-alternative}

Several remedies may be considered to address the unraveling results of Theorems~\ref{the-main-result} and~\ref{the-witness}. First, agents could be asked to investigate and report their findings simultaneously, a structure that we may call {\em parallel monitoring}. Second, agents could investigate past investigators. Third, in some cases, such as criminal cases, agents can in principle be incentivized by the perspective that their findings have an influence of subsequent crimes or events. Finally, the same agents could be called to make statements repeatedly over time. These possibilities are examined in turn.

{\bf 1.	Parallel monitoring and weak implementation with a centralized authority}

In applications such as historical revisionism, it is realistic to assume that agents make their statements sequentially. It is nonetheless natural to consider, from a mechanism-design perspective, the case in which several agents simultaneously and independently investigate the fact of interest and report their findings to some central authority. For example, blood-testing laboratories could be asked to test their respective samples independently from each other and report their findings simultaneously to some overseeing agency. The laboratories would be rewarded if their findings are match and punished if they don't.\footnote{The academic refereeing process has reporting features that resemble the parallel-monitoring design, although the incentives for referees are different and arguably more complex than a mere coordination motive.}

An immediate concern with this solution stems from the incentives of the central authority. If the authority has a material interest in a specific outcome (which cannot be ruled out, especially in politically charged investigations), it can secretly help agents coordinate on some report or influence the agents' reports in various ways. In order to avoid this, the central agency must be itself monitored, which brings us back to the sequential monitoring problem. In the blood-testing example, if laboratories must report their findings simultaneously, there is no recourse for an athlete accused of doping, particularly if laboratories can coordinate their response.

In some cases, parallel monitoring may be difficult to implement: it may be difficult, for instance, to send multiple investigators on a crime scene to independently interrogate witness and collect evidence, without the investigators being able to communicate, either directly or through witnesses and possibly coordinate. Moreover, when the evidence is limited supply (such as the weapon of a crime), such a limitation creates negative correlation in investigators' reports, since at most one of them can discover the evidence (weapon).\footnote{The deleterious impact of negative correlation on the informativeness of agents with a coordination motive is a central finding in Pei and Strulovici's (2020) analysis of strategic abuse.}

Finally, successful parallel monitoring can coexist only with other equilibria, many of which are uninformative. Indeed, there always exist equilibria in which monitors coordinate on a predetermined sequence of reports. Among these equilibria, successful parallel monitoring, when it is feasible, may be non-robust. Even small amounts of strategic uncertainty about other monitors can suffice to destroy the informative equilibrium.\footnote{This idea is explored in ongoing joint work with Harry Pei.}

{\bf 2.	Monitoring the Monitor}

In the applications discussed so far, each agent makes a statement about a particular question (e.g., whether an athlete used a banned substance, or who committed a crime). The agents themselves could be investigated. In this case, the subject of study of each agent evolves over time: agent 1 investigates the initial question, agent 2 investigates agent 1's treatment of the initial question, agent 3 investigates agent 2's investigation and so on. Such a sequence is called a ``monitoring hierarchy.'' In his analysis of monitoring, Hurwicz (2007) considers a similar hierarchy, except that  the number of agents is finite and the monitoring chain cycles repeatedly across the a fixed number of agents. A first conceptual difficulty with this approach concerns the simultaneity and complexity of these monitoring tasks: agents are supposed to conduct an infinite amount of monitoring tasks and are indirectly the subject of the tasks that they are investigating.

Even if we allow an infinite sequence of distinct agents, each of which is tasked with investigating the previous agent in the sequence, another issue emerges: what if an agent who discovers evidence about the agent he was monitoring can hide or destroy the evidence in exchange for a payment from the guilty agent? Such a transfer amounts to a local form of corruption among nearby agents. In a separate paper, I show that even this local form of corruption may suffice to unravel agents' incentives (Strulovici (2019)). Intuitively, the tension comes from two considerations. First, in order for a monitor to accurately report his finding about the previous agent, the monitor’s punishment if he lies and makes a wrongful accusation must be higher than his reward for claiming to have discovered wrongdoing by the previous monitor. Second, a monitor must get a higher reward if he finds evidence against the agent that he was investigating, relative to the punishment that this agent gets when punished. Otherwise, there is a transfer form the guilty agent to the monitor that improves both agents' payoff. Combined, these observations imply that the chain of rewards is unbounded and, hence, unsustainable.

{\bf 3.  Repeated Setting}

When mediated learning concerns the identification of a criminal and opportunities to commit crime are repeated over time, it is a priori possible that investigators care about the truth indirectly, through the impact that their findings have on citizens' future behavior.

Suppose that a citizen's decision to commit crime depends on whether his past actions were accurately called by past investigators: for example, a citizen who was wrongfully accused of committing crime in the past or who was wrongfully acquitted for crimes that he did commit may be more likely to commit crime in the future. In this setting, investigators could in principle have an endogenous incentive to report accurate findings. Provided that players are sufficiently patient, this kind of strategy profile could a priori be used to incentivize accurate reporting.

However, for this argument to work, a citizen’s strategy must depend on his private history, where public history consists of official findings about citizens' past actions and a citizen's private history records his actual past actions. In order for a citizen's private history to affect his decision of whether to commit crime, the citizen must be made indifferent between committing crime and abstaining from it, a knife-edge condition that is violated if, for instance, citizens are subject to small private shocks affecting their benefits from committing crime.\footnote{The argument is somewhat similar to the section on witnesses in this paper.}

{\bf 4.  Alternating Statements}

Finally, one could ask a fixed set of agents to take turns investigating and reporting on the question of interest. The key difference with the baseline model is that agents now have private information about what they did in the past, which affects how they interpret the declarations of other agents. If an agent has discovered disposable signals in the past, he knows that other agents have fewer signals to discover. The analysis becomes more complex because agents' decisions now depend on their beliefs about the amount of evidence left, about other agents' beliefs about the amount evidence, their belief about agents' beliefs about their beliefs about the amount of evidence left and so on. While information attrition is likely to have a similar effect as in this paper’s model, confirming this intuition and exploring this question is left for future research.

\subsection{Designing Truth-Dependent Preferences: Oaths, Capitalism, and Popular Juries}\label{ssec-oaths}

A more direct approach to improving mediated learning is to increase the salience of truth-dependent preferences.

This may be achieved by fostering agents' ethical sense, from inculcating an ethical education and culture to strengthening trust in institutions and developing effective vetting and selection processes for key learning responsibilities.

Professional oaths, from the Hippocratic oath in medicine journalistic oaths such as Walter Williams' Journalist's Creed (Farrar (1998)) also aim at eliciting ethical behavior. This paper suggests a positive correlation between the need of oaths in various professions and the severity of information attrition problems in these professions.

Even if a limited fraction of agents is swayed by such oaths, this may in principle suffice for incentivizing truthful behavior by other agents. Studying the mechanisms and behavioral features through which ethical agents can incentivize mediated learning is beyond the scope of this paper, but it is easy to conceive of simple examples: suppose that an agent, who is known to truthfully seek and report the truth is commonly know to appear in round $N > 1$. This agent's report provides reliable information, akin to an exogenous public signal about the state of the world, which can be used to incentivize all agents coming in rounds $i<N$.\footnote{Even if agent $N$ has only a small probability $p <1$ of behaving ethically, his report may still be used to incentivize agents in earlier around as long as these agents' rewards and punishments are of order $1/p$.}

Another approach is to design society in a way that increases information mediators' material dependence on the truth, i.e., gives them ``skin in the game.'' Eliciting information from agents about a scientific fact or the social value of a new product or process is easier when the agents stand to gain financially from this information, which may broadly interpreted as capitalistic incentives. Thus interpreted, the theory offers a new perspective on the ``virtue'' of capitalism relative to systems in which agents have low-powered incentives.\footnote{A large literature emphasizes capitalism's ability to reduce moral hazard problems relative to socialistic systems. See Myerson (2007) and Tirole (2006) for a review of relevant papers and corporate-finance models capturing this idea. The question of incentive compatibility and its relation to various economic systems is at the heart of Leonid Hurwicz's development of mechanism design (Hurwicz (1973)). Compared to the literature, this paper is concerned with both moral hazard and adverse selection since the object of mediated learning is to elicit effort from
investigators as well as truthful revelation of investigators' findings.} The theory also emphasizes that violations of the rules of capitalism (or of any system, for that matter) may be difficult to detect and reveal truthfully, and thus suggests a potential tradeoff between the incentives provided within a given a economic or political system and the incentives required to guarantee that the system itself is obeyed.

A final angle to attack mediated learning failures is to democratize the learning process by enlarging the pool of potential information intermediaries. Large pools can increase the alignment---real or perceived---between the intermediaries and society as a whole, in contrast to the baseline model of the paper, in which society's objective is dissociated from the intermediaries'. Large pools of intermediaries are conceivable when the expertise required to learn the fact of interest is limited. The institution of popular juries may be viewed as one application, which trades off intermediaries' expertise with their representativeness of a more global and diffuse body of stakeholders.

\newpage

\appendix

\setlength{\parindent}{0pt}

\section{Proof of Theorem~\ref{the-main-result}}\label{sec-proof-th1}

For any $i,k\geq 1$, let $F^k_i = Pr(|S_i| \geq k\enskip |\enskip m_1^{i-1})$ denote the probability that there remain at least $k$ signals to discover at the beginning of round $i$ given past reports $m_1^{i-1}$. The prior probability $F^k$ that $S$ contains at least $k$ signals at the beginning of the investigation process satisfies $F^k= F^k_1$.

\subsection{Preliminary Results}\label{ssec-prelims}

Agents' decisions are invariant with respect to a uniform translation in their gross utility (i.e., the utility that they get before taking into account their effort cost). Therefore, we can assume without loss of generality that their gross utility takes values in some interval $[0,R]$. Moreover, since $R$ is an upper bound on payoffs, it can always be increased to guarantee that $R > c$, which we will assume without loss of generality.

%Given some equilibrium and integers $k$ and $i$, let $f^k_i$ denote the equilibrium probability that there remain $k$ signals to discover conditional on the reports $m_1^{i-1}$---the dependence on $m_1^{i-1}$ is omitted from the notation for simplicity. Recalling that $F^k_i$ denotes the probability that there remains at {\em least} $k$ signals given $m_1^{i-1}$, we have$F^k_i = \sum_{k'\geq k} f^k_i$.

%Given any equilibrium and history $m_1^{i-1}$, we can view the continuation equilibrium from round $i$ onwards as an equilibrium of a modified environment where the initial beliefs are described by the conditional distribution given $m_1^{i-1}$ and the compensation functions are defined for all $j\geq i$ by $\tilde V_j(m_i, m_{i+1},\ldots) = V_j(m_1^{i-1}, m_i,\ldots)$. In the statement to follow, an ``equilibrium'' means a continuation equilibrium continuation equilibrium with an (informative) equilibrium, by resetting the index.

Let $\gamma_i$ denote the probability that $i$ works given $m_1^{i-1}$.

\begin{lemma}\label{lem-start} $\gamma_i > 0$ only if $F^1_i\geq \frac{c}{R} > 0$.\end{lemma}

\begin{proof} Let $V^w_i$ denote $i$'s expected gross utility if he works. $V^w_i$ may be decomposed in terms of $i$'s expected gross utility $\bar V_i$ if he works {\em and} there exists some signal left to be found, and his expected gross utility $V^{w,\ep}_i$ if he works but there is no signal left to be found ($S_i = \ep$):
\[V^w_i = F^1_i \bar V_i + (1-F^1_i) V^{w,\ep}_i.\]
Conditional on $S_i=\ep$, $i$'s expected gross utility if he works is the same as his expected gross utility $V^{s,\ep}_i$ if he shirks and uses the same reporting strategy as he does after working and finding nothing: conditional on $i$'s report (whatever it is), the distribution of reports by subsequent agents is identical since there is no signal left to be found. Therefore, $V^{w,\ep}_i = V^{s,\ep}_i$. Furthermore, we also have $\bar V_i\leq R$ since $R$ is the maximum possible gross utility.

Therefore, $i$'s net utility $U^w_i$ from working, including the cost of working, satisfies
\[U^w_i\leq F^1_i R + (1-F^1_i) V^{s,\ep}_i- c.\]
Similarly, $i$'s utility $U^s_i$ from shirking satisfies
\[U^s_i\geq F^1_i \times 0 + (1-F^1_i) V^{s,\ep}_i = (1-F^1_i) V^{s,\ep}_i\]
where the inequality comes from the fact that 0 is a lower bound on $i$'s realized gross utility. Comparing the previous two equations shows that shirking strictly dominates working if $F^1_i R - c < 0.$\hfill$\blacksquare$
\end{proof}

Given any round $i$ and report history $m_1^{i-1}$ such that $\gamma_i >0$, let
\bit
\vspace{-.4cm}
\item $V^*_i$ denote $i$'s maximal expected gross utility if he shirks, where the maximum is taken over all possible messages $m_i\in M$;
\item $f^0_i = 1-F^1_i$ denote the probability that $S_i = \ep$ at the beginning of round $i$.
\eit

\begin{lemma}\label{lem-compensation-empty} Suppose that $\lambda < 1$. If $i$ works and finds nothing and sends message $m_i$, his expected gross utility  $V^{w}_i(\ep,m_i)$ satisfies
\[V^{w}_i(\ep,m_i) \leq V^*_i + \frac{f^0_i R}{1-\lambda}\]
for all $m_i$.
\end{lemma}

% Rm: G is independent of k

\begin{proof} For any sequence $S''$ of signals, let $\Delta_i(S'')$ denote the probability that $S_i =S''$ conditional on report history $m_1^{i-1}$ and $\Delta_i^\ep(S'')$ denote the probability that $S_i = S''$ conditional on $i$ working and finding nothing. Bayesian updating implies that, for any $S''\neq \ep$,
\[\Delta_i^\ep(S'') = \Delta_i(S'') \frac{(1-\lambda)}{(1-f^0_i)(1-\lambda) + f^0_i}\]
and, for $S'' = \ep$,
\[\Delta_i^\ep(\ep) = \Delta_i(\ep)\frac{1}{(1-f^0_i)(1-\lambda) +f^0_i}.\]
This implies that
\be\label{eq-new-delta} \Delta_i^\ep(S'') - \Delta_i(S'') = - \frac{\lambda f^0_i \Delta_i(S'')}{(1-f^0_i)(1-\lambda) + f^0_i}\ee
for $S''\neq \ep$ and
\be\label{eq-new-delta-ep} \Delta_i^\ep(\ep) - \Delta_i(\ep) = \frac{\lambda (1-f^0_i) \Delta_i(\ep)}{(1-f^0_i)(1-\lambda) + f^0_i}.\ee
Let $V_i(m_i,S'')$ denote $i$'s expected gross utility conditional on $i$ producing evidence $m_i$  and on $S_{i+1} = S''$. Notice that $m_1^i = (m_1^{i-1}, m_i)$ and $S_{i+1}$ completely determine the distribution of reports $\{m_j\}_{j>i}$. Therefore, $V_i(m_i,S'')$ is the same regardless of whether $i$ has worked or shirked. Agent $i$'s expected gross utility conditional on i) working, ii) finding no signal, and iii) producing message $m_i$, is
\[V_i^w(\ep, m_i) = \sum_{S''\in \mc S} V_i(m_i,S'') \Delta_i^\ep(S''),\]
whereas his expected gross utility if $i$ shirks and sends message $m_i$ is
\[V_i^s(m_i) = \sum_{S''\in \mc S} V_i(m_i,S'') \Delta_i(S'')\]
because $i$ has learned nothing from shirking and thus holds the same belief as his prior belief at the beginning of round $i$. Combining these expressions, we get \be\label{eq-diff-V-delta}V_i^w(\ep, m_i) - V_i^s(m_i) = \sum_{S''\in\mc S} V_i(m_i, S'') (\Delta_i^\ep(S'') - \Delta_i(S'')).\ee
Since $V_i(m_i, S'')\in [0,R]$ for all $m_i$ and $S''$, combining~\eqref{eq-diff-V-delta} with~\eqref{eq-new-delta} and~\eqref{eq-new-delta-ep} yields
\[V_i^w(\ep, m_i) - V_i^s(m_i) \leq \frac{R\Delta_i(\ep)\lambda(1-f^0_i)}{(1-f^0_i)(1-\lambda) +f^0_i}.\]
Since $\lambda < 1$, the denominator is bounded below by $1-\lambda$. Since $\Delta_i(\ep) = f^0_i$, the numerator is bounded above by $Rf^0_i$. This yields
\[V_i^w(\ep, m_i) \leq V_i^s(m_i) + f^0_i \frac{R}{1-\lambda}\leq V^*_i + f^0_i \frac{R}{1-\lambda},\]
which proves the lemma. Intuitively, this results means that if $f^0_i$ is negligible relative to $(1-\lambda)$, then $i$'s expected gross utility after working and finding nothing cannot be much higher than if $i$ had shirked, because finding nothing in this case merely reveals that $i$ was unlucky and otherwise conveys little else information.\hfill$\blacksquare$\end{proof}

For each round $i$ and subset $\tilde M_i\subset M$ of messages, we introduce several probabilities conditional on $m_1^{i-1}$.
\bit
\vspace{-.4cm}
\item $\gamma_i(\tilde M_i)$: probability that $i$ produces a report in $\tilde M_i$ conditional on working;
\item $g_i(\tilde M_i)$: probability that $i$ finds a signal and produces a report in $\tilde M_i$ conditional on working;
\item $d_i(\tilde M_i)$: probability that $i$ finds no signal and produces a report in $\tilde M_i$ conditional on working.
\eit
Finally, let $M_i^+$ denote the set of messages $m_i$ that are followed by an informative continuation equilibrium at round $i+1$.

\begin{lemma}\label{lem-compensation-bound-mixed} Agent $i$'s expected gross utility conditional on working has the following upper bound. If $\lambda < 1$, then
\[V^w_i\leq V^*_i + d_i(M^+_i) \frac{f^0_i R}{1-\lambda}+ g_i(M^+_i)R.\]
If $\lambda = 1$, then
\[V^w_i\leq V^*_i + d_i(M^+_i)  R + g_i(M^+_i) R.\]
\end{lemma}

\begin{proof}
For each $m_i\in M$, let $V^w_i(m_i)$ denote $i$'s expected gross utility conditional on working and sending message $m_i$ and $M^-_i$ denote the set of messages $m_i$ after which no $j>i$ ever works, so that $M = M^+_i \cup M^-_i$ and $M^+_i\cap M^-_i = \ep$.

We have
\be\label{eq-newCw}V^w_i = \sum_{m_i\in M^-_i} \gamma_i(m_i) V^w_i(m_i) + \sum_{m_i\in M^+_i} \gamma_i(m_i) V^w_i(m_i).\ee

For the first term, note that $i$'s expected utility conditional on reporting $m_i$ and on no $j>i$ ever producing real evidence does not depend on whether $i$ worked or shirked: either way, the distribution of the reports $\{m_j\}_{j>i}$ is independent of the set of signals that remain in the case. Letting, as in the previous lemma, $V^s_i(m_i)$  denote $i$'s expected gross utility conditional on shirking and sending message $m_i$, we thus have $V^w_i(m_i) = V^s_i(m_i)$ for all $m_i\in M^-_i$. Since $V^*_i = \max_{m_i\in M} V^s_i(m_i)$, the first term in~\eqref{eq-newCw} is bounded above by $\gamma_i(M_i^-) V^*_i$.

For the second term, we have $\gamma_i(m_i) = d_i(m_i) + g_i(m_i)$ and
\[\gamma_i(m_i) V^w_i(m_i) \leq d_i(m_i) V^w_i(\ep, m_i) + g_i(m_i) R,\]
where we used the fact that $i$'s expected gross utility conditional on working, finding a signal, and reporting $m_i$ is bounded by $R$.

Combining these observations yields
\be\label{eq-newCw-2}V^w_i \leq \gamma_i(M_i^-) V^*_i + g_i(M_i^+) R + \sum_{m_i\in M^+_i} d_i(m_i) V^w_i(\ep, m_i).\ee

If $\lambda < 1$, Lemma~\ref{lem-compensation-empty} implies that $V^w_i(\ep,m_i)\leq V^*_i + \frac{f^0_i R}{1-\lambda}$. Summing over all $m_i\in M^+_i$, we get
\be\label{eq-delta-term}\sum_{m_i\in M^+_i} d_i(m_i) V^w_i(\ep, m_i)\leq d_i(M_i^+) V^*_i + d_i(M^+_i)\frac{f^0_i R}{1-\lambda}.\ee
Since $\gamma_i(M^-_i) + d_i(M^+_i) \leq \gamma_i(M^-_i) + \gamma_i(M^+_i)  = 1$, combining~\eqref{eq-newCw-2} and \eqref{eq-delta-term} proves the lemma when $\lambda < 1$.

If $\lambda =1$, using in~\eqref{eq-newCw-2} the fact that $V^w_i(\ep,m_i)$ is bounded above by $R$ directly proves the lemma.\hfill$\blacksquare$\end{proof}

For each round $i$, we define another set of probabilities conditional on $m_1^{i-1}$:
\bit
\vspace{-.4cm}
\item $\beta_i$: probability that $i$ discovers a signal (before observing whether $i$ works, i.e., viewed from the beginning of round $i$);
\item $\beta_i(\tilde M_i)$: probability that $i$ discovers a signal and sends a message in $\tilde M_i$;
\item $\alpha_i(\tilde M_i)$: probability that $i$ shirks and sends a message  in $\tilde M_i$;
\item $\delta_i(\tilde M_i)$: probability that $i$ produces a report in $\tilde M_i$ conditional on working {\em and} finding no signal;\footnote{Note that $\delta_i(\tilde M_i)\geq d_i(\tilde M_i)$, where $d_i(\tilde M_i)$ was defined before Lemma~\ref{lem-compensation-bound-mixed}.}
\item $F^k_{i+1}(m_i)$: probability that there remain at least $k$ signals at the beginning of round $i+1$ given reports $m_1^i =(m_1^{i-1}, m_i)$.
\eit

% REM: we already used C() for the compensation!!!

\begin{proposition}\label{pro-general-mixed} Let $C(\lambda) = 2R/(c\lambda(1-\lambda))$ for $\lambda <1$ and $C(1) = 2R/c$. The following inequality holds for all constants $C\geq C(\lambda)$, round $i$, and integer $k\geq 1$ such that $F^k_i > C F^{k+1}_i$:
\be\label{eq-saver}\beta_i \leq C \frac{\mb E_i\left[(F^k_i - F^k_{i+1}(m_i))1_{m_i\in M^+_i}\right]}{F^k_i - C F_i^{k+1}}.\ee
\end{proposition}
\begin{proof} The right-hand side of~\eqref{eq-saver} is increasing in $C$ over the range of $C$ that satisfy the condition $F^k_i - C F_i^{k+1}>0$. Therefore, if~\eqref{eq-saver} is satisfied for any $C$ such that $F^k_i - C F_i^{k+1}>0$, it is also satisfied for any $C'\geq C$ such that $F^k_i - C' F_i^{k+1}>0$. The proposition thus follows if we show the inequality for $C(\lambda)$.

First, we show that the claim holds if $\beta_i =0$. In this case, $i$ must shirk with probability 1 (if not, Lemma~\ref{lem-start} implies that $S_i$ is nonempty with positive probability and, hence, that $\beta_i>0$). Since~$i$ shirks with probability 1, there is no belief update between rounds $i$ and $i+1$. Therefore, $F^k_i = F^k_{i+1}(m_i)$ for any message $m_i$ that $i$ sends in equilibrium. The right-hand side of~\eqref{eq-saver} is thus equal to zero and~\eqref{eq-saver} is satisfied.

Now suppose that $\beta_i >0$ or, equivalently, that $\gamma_i > 0$: $i$ works with positive probability. We consider two cases, distinguished by whether working is likely to lead to the discovery of a signal and, under $i$'s equilibrium strategy, to a message in $M^+_i$.

{\bf Case 1: $g_i(M^+_i) \geq \frac{c}{2R}$.} By definition, we have
\[\beta_i = \gamma_i \lambda(1-f^0_i)\]
(in particular, $\beta_i \leq \gamma_i$) and
\[\beta_i(M^+_i) = \gamma_i g_i(M^+_i).\]
Since $g_i(M^+_i)\geq c/2R$, this implies that
\be\label{eq-lambda-scale}\beta_i \leq \gamma_i\leq \frac{2R}{c} \beta_i(M^+_i).\ee
Therefore,~\eqref{eq-saver} will follow if we prove that $\beta_i(M^+_i)$ is bounded above by$\frac{\mb E_i\left[(F^k_i - F^k_{i+1}(m_i))1_{m_i\in M^+_i}\right]}{F^k_i - C F_i^{k+1}}$ for $C = 2R/c$.

For each $m_i$ that $i$ may send in equilibrium and $k\geq 1$, Bayesian updating implies that
\be\label{eq-bayes-nonempty-mix} F^k_{i+1}(m_i) = \frac{F^k_i\left(\alpha_i(m_i) + \gamma_i(1-\lambda)\delta_i(m_i)\right)+ \Phi(m_i)}{\alpha_i(m_i) + \gamma_i((1-F^1_i + F^1_i(1-\lambda))\delta_i(m_i) + \beta_i(m_i)},\ee
where $\Phi(m_i)$ is the probability that i) $i$ works, ii) $i$ discovers a signal, iii) $i$ sends report $m_i$, and iv) there remain at least $k$ signals at the beginning of round $i+1$.

Let $p_i(m_i)$ denote the probability that $i$ produces report $m_i$ conditional on $m_1^{i-1}$: $p_i(m_i)$ is the denominator of~\eqref{eq-bayes-nonempty-mix}. Rearranging~\eqref{eq-bayes-nonempty-mix} and simplifying, we have
\be\label{eq-intermediate-new-use}F^k_i(\beta_i(m_i) + \gamma_i (1-F^1_i) \lambda \delta_i(m_i)) = (F^k_i - F^{k+1}_i(m_i))p_i(m_i) + \Phi(m_i)\ee
Since $\gamma_i (1-F^1_i) \lambda \delta_i(m_i)\geq 0$, summing the previous equation over $m_i\in M_i^+$ yields
\be\label{eq-intermediate-beta-bound} F^k_i \beta_i(M_i^+) \leq \mb E_i\left[\left(F^k_i - F^k_{i+1}(m_i)\right) \ind_{m_i\in M_i^+}\right] + \sum_{m_i\in M^+_i} \Phi(m_i).\ee

%In words, the posterior probability that there are at least $k$ pieces left given $m_i$ comes from two components: either $m_i$ was fabricated, in which case we stick to the prior $F^k_i$, or it is real, and in this case there are at least $k$ pieces left only if there were at least $k+1$ pieces left at the beginning of round $i$, the latter event having probability $F^{k+1}_i$ as viewed from the beginning of round~$i$.

Since $\Phi(m_i) = \mb E_i\left[\ind_{|S_i |\geq k+1} \enskip \ind_{\mbox{i works, discovers a signal, and reports $m_i$}}\right]$, we have
\begin{align}\label{eq-bounding-the-above-terms-mixed}
\sum_{m_i\in M^+_i} \Phi(m_i) & \leq \sum_{m_i\in M^+_i}  \mb E_i\left[\ind_{|S_i |\geq k+1} \enskip \ind_{\mbox{i works and reports $m_i$}}\right] \nonumber
\\ & =  \mb E_i\left[\ind_{|S_i |\geq k+1} \enskip \ind_{\mbox{i works and reports $m_i\in M^+_i$}}\right]\nonumber
\\ &\leq \mb E_i\left[\ind_{|S_i |\geq k+1} \enskip \ind_{\mbox{i works}}\right]\nonumber \\ & = F^{k+1}_i\gamma_i,\end{align}
noting, for the last equality, that the event that $i$ works, which has probability $\gamma_i$, depends only on $m_1^{i-1}$ and is thus independent of the event $\{|S_i|\geq k+1\}$ conditional on $m_1^{i-1}$.

Combining this with~\eqref{eq-intermediate-beta-bound} yields
\be\label{eq-important-for-case2}F^k_i \beta(M^+_i)\leq \mb E_i\left[(F^k_i - F^k_{i+1}(m_i))1_{m_i\in M^+_i}\right] + F^{k+1}_i \gamma_i.\ee
Since $g_i(M^+_i)\geq c/2R$, we have $\beta(M^+_i) = \gamma_i g_i(M^+_i)\geq \gamma_i c/2R$. Inequality~\eqref{eq-important-for-case2} then yields
\be\label{eq-intermediate-beta-big}\beta(M^+_i) \leq \frac{1}{F^k_i - 2R/c F^{k+1}_i} \mb E_i[(F^k_i - F^k_{i+1}(m_i))1_{m_i\in M^+_i}].\ee
Combining this with~\eqref{eq-lambda-scale} yields~\eqref{eq-saver} for $C(1) = 2R/c$.\footnote{Note that the proposition's assumption that $F^k_i - C(\lambda) F^{k+1}_i>0$ implies that $F^k_i - \frac{2R}{c} F^{k+1}_i>0$ since $C(\lambda)\geq C(1)$ regardless of $\lambda$.} Since $C(\lambda)\geq C(1)$ for all $\lambda \in (0,1]$, the monotonicity noted at the beginning of the proof yields the desired conclusion for $C(\lambda)$.

{\bf Case 2: $g_i(M^+_i) < \frac{c}{2R}$.} We prove that $\gamma_i$ is bounded above by the right-hand side of~\eqref{eq-saver} for $C = C(\lambda)$. Since $\gamma_i\geq \beta_i$, this will yield the desired conclusion.

Intuitively, in Case 2 the probability of discovering a signal that, together with $i$'s equilibrium message strategy, triggers subsequent work is too low to incentivize $i$ to work. The only way of incentivizing $i$ to work is therefore for him to signal by his message that he found nothing through work. For this to happen, the probability $f_i^0$ that there remains no evidence must be high enough. We will us this fact to obtain a bound on $\gamma_i$.

From Lemma~\ref{lem-compensation-bound-mixed}, if $g_i(M^+_i) < c/2R$, $i$'s utility from working is bounded above by
\[U^w_i = V^w_i-c \leq V^*_i + d_i(M_i^+)\frac{f^0_i R}{1-\lambda} - \frac{c}{2}\]
if $\lambda <1$, and by
\[U^w_i\leq V^*_i + d_i(M_i^+) R - \frac{c}{2}\]
if $\lambda = 1$. Therefore, working is optimal only if $d_i(M_i^+) f^0_i\geq c(1-\lambda)/2R$ when $\lambda <1$ and only if $d_i(M_i^+) \geq c/2R$ when $\lambda = 1$.

%For any fixed value of $\lambda$, this yields a lower bound $\lbar f >0$ for $d_i(M_i^+) f^0_i$.
%Thus suppose that there is a constant $\lbar f > 0$ such that $d_i(M_i^+) f^0_i >\lbar f$.

Summing~\eqref{eq-intermediate-new-use} over $M_i^+$ and using~\eqref{eq-bounding-the-above-terms-mixed} and $f^0_i = 1-F^1_i$ yields
\be\label{eq-int-di-delta}F^k_i \beta_i(M^+_i) +\gamma_i F^k_i \delta_i(M^+_i) f^0_i \lambda \leq \mb E_i\left[(F^k_i - F^k_{i+1}(m_i))1_{m_i\in M^+_i}\right] + F^{k+1}_i \gamma_i.\ee

For $\lambda < 1$, we have $d_i(M_i^+) f^0_i\geq c(1-\lambda)/2R$. Since $\delta_i(m_i)\geq d_i(m_i)$ for all $m_i$ (by definition) and  $F^k_i \beta_i(M^+_i)\geq0$,~\eqref{eq-int-di-delta} implies that
\[\gamma_i \leq \frac{\mb E_i\left[(F^k_i - F^k_{i+1}(m_i))1_{m_i\in M^+_i}\right]}{F^k_i c \lambda (1-\lambda)/2R - F^{k+1}_i}.\]
Multiplying the numerator and denominator of the right-hand side by $C(\lambda)$ yields the result.

For $\lambda =1$, we have $d_i(m_i) = \delta_i(m_i) f^0_i $ for all $m_i$ and hence $\delta_i(M^+_i) f^0_i \lambda = d_i(M^+_i)$, which is greater than $c/2R$  as noted earlier. Therefore,~\eqref{eq-int-di-delta} implies that
\[\gamma_i \leq \frac{\mb E_i\left[(F^k_i - F^k_{i+1}(m_i))1_{m_i\in M^+_i}\right]}{F^k_i (c/2R) - F^{k+1}_i}.\]
Multiplying the numerator and the denominator of the right-hand side by $C(1)$ yields the result.\hfill$\blacksquare$
\end{proof}

\subsection{Proof of Theorem~\ref{the-main-result}}\label{ssec-proof}

The proof proceeds by induction on $k$. Lemma~\ref{lem-start} already proves the claim for $k=1$ with $\lbar F^1 = c/R$. Now suppose that the claim holds for some $k\geq 1$: there exists a threshold $\lbar F^k > 0$ such that any informative continuation equilibrium in round $i$ satisfies $F^k_i\geq \lbar F^k$.\footnote{Note that while the theorem is stated at round 0, it applies equally well to all continuation equilibria since the thresholds $\{\lbar F^k\}_{k\geq 1}$ depend only on the parameters $(R, P, c,\lambda)$, which are constant throughout the game.} We will show that there is a constant $\lbar F^{k+1}>0$ such that the continuation equilibrium can be informative only if $F^{k+1}_i\geq \lbar F^{k+1}$.

For any $\omega\in[0,1]$, let $F^k(\omega) = \inf\{F^k_i\in [0,1]: \gamma_i>0, F^{k+1}_i\leq \omega\}$ where the infimum is taken over all on-path histories and all equilibria of the game, and is by convention equal to 1 if no informative equilibrium exists for which $F^{k+1}_i\leq \omega$. Note that $F^k(\omega)$ is by construction nonincreasing in $\omega$.

Let $\lbar F = \inf\{\omega: F^k(\omega) < 1\}$. $\lbar F$ is the smallest value of $F^{k+1}_i$ for which an informative continuation equilibrium exists (or the infimum of such values, if the infimum is not reached).

Our objective is to show that $\lbar F > 0$. Let $\bar F^k = \lim_{\omega \downarrow \lbar F} F^k(\omega)$, i.e., the right limit of $F^k(\cdot)$ at $\lbar F$. This limit is guaranteed to exist because $F^k(\cdot)$ is nonincreasing. Roughly put, $\bar F^k$ is the smallest probability of having at least $k$ signals among informative equilibria that have the smallest possible probability having at least $k+1$ signals left.

Also let $\ve > 0$ denote any small constant such that $\hat F^k = F^k(\lbar F + G \ve)\geq \frac{\bar F^k}{1+\eta}$, where $G>0$ is a large constant and $\eta>0$ is a small constant determined at the end of the proof independently of $k$.\footnote{For example, one can choose $G$ and $\eta$ so that $\sqrt{G} = \max\{64 R^3/(\lambda c^3), 12 R/c\}$ and $\eta = 1/\sqrt{G}$, as explained at the end of this proof.} Since $\bar F^k$ is the right-limit of $F^k(\cdot)$ at $\lbar F$, such an $\ve$ exists. Moreover, since all $\ve'\in(0,\ve)$ also satisfy the condition, we can choose $\ve$ such that
\be\label{eq-bound-ve-mixed} \ve\leq \frac{1}{2}\left(\frac{\hat F^k}{G}\right)^2.\ee

By definition of $\lbar F$, there exists at least one informative continuation equilibrium for which $F^{k+1}_i\in [\lbar F, \lbar F +\ve]$. Moreover, by definition and monotonicity of $F^k(\cdot)$, there exists an informative continuation equilibrium among those for which $F^k_i \leq \bar F^k + \eta \hat F^k$.

Consider such a continuation equilibrium, which starts at round $i$, say, and suppose by contradiction that
\be\label{eq-induction-hypo} \lbar F  = 0.\ee

Let $\mc A$ denote the event that $F^{k+1}_j\leq G \ve$ for all $j\geq i$.

\begin{lemma}\label{lem-supermart} $i$ assigns probability at least $1-1/G$ to $\mc A$.\end{lemma}
\begin{proof} For $j\geq i$, let $\bar F^{k+1}_j$ denote the probability assigned at the beginning of round $j$ (i.e., conditional on $m_1^{j-1}$) that there remain at least $k+1$ signals {\em at the beginning of round $i$} (fixed). The process $\{\bar F^{k+1}_j\}_{j\geq i}$  is nonnegative and bounded above by 1, and it is a martingale by the law of iterated expectations and the fact that $j$'s filtration grows finer as $j$ increases. Moreover, $\bar F^{k+1}_i = F^{k+1}_i\leq \ve$.

Doob's martingale inequality therefore implies that for any $J\geq i$, $Pr(\max_{i\leq j\leq J} \bar F^{k+1}_j\geq G\ve)\leq \frac{\mb E_i[\bar F^{k+1}_i]}{G\ve} \leq 1/G$. The event $\bar{\mc A}_\infty$ defined by $\{\max_{j\geq i} \bar F^{k+1}_j\leq G\ve\}$ is the intersection of the events $\bar{\mc A}_J = \{\max_{i\leq j\leq J} \bar F^{k+1}_j\leq G\ve\}$. Therefore, $Pr(\bar{\mc A}_\infty) = \lim_{J\rightarrow\infty} Pr(\bar{\mc A}_J) \geq 1-1/G$.

Finally, we note that $\bar F^{k+1}_j \geq F^{k+1}_j$ for all $j\geq i$, because the true number of remaining signals only decreases over time and thus whatever signals remained at the beginning of round $i$ must have contained the signals that remain at the beginning of round $j$, so $Pr(\mc A)\geq Pr(\bar{\mc A}_\infty)\geq 1-1/G$.\hfill$\blacksquare$
\end{proof}

Conditional on $\mc A$, $F^{k+1}_j\leq G\ve$ for all $j\geq i$. Moreover, if round $j$ belongs to an informative continuation equilibrium, we have $F^k_j\geq F^k(F^{k+1}_j)\geq F^k(G\ve) = \hat F^k$. Given any positive constant $C$, this implies that for informative continuation equilibrium starting in round $j$,
\be\label{eq-denominator-bound}F^k_j - C F^{k+1}_j \geq \hat F^k - C G\ve \geq \hat F^k/2  >0\ee
provided that $G\geq C$, where the last weak inequality comes from~\eqref{eq-bound-ve-mixed}. The condition $G\geq C$ will be satisfied by fixing $C$ at the level $C(\lambda)$ stated in Proposition~\ref{pro-general-mixed} and then choosing $G$ large enough. (A specific value of $G$ is given at the end of the proof.)

Let $\mc A_j$ denote the event that $F^{k+1}_l\leq G\ve$ for all integers $l$ such that $i\leq l\leq j$. The events $\{\mc A_j\}_{j\geq i}$ form a decreasing sequence that converges to $\mc A$ as $j\rightarrow\infty$. Moreover, $\mc A_j$ is measurable with respect to the information available at the beginning of round $j$.

From Proposition~\ref{pro-general-mixed}, we have
\be\label{eq-big-brackets}\beta_j\leq \mb E_j\left[\frac{C(F^k_j - F^k_{j+1}(m_j))\ind_{m_j\in M^+_j}}{F^k_j - C F_j^{k+1}}\right]\ee
over $\mc A_j$.

Let $\mc B_j$ denote the event that $m_j\in M^+_j$. $\mc B_j$ is adapted to $m_1^j$, i.e., known at the beginning of round $j+1$. Notice that if $\mc B_j$ does {\em not} occur, it means by definition that no $l>j$ ever works (i.e., $m_1^{j}$ is followed by an {\em un}informative continuation equilibrium). This implies that the sequence of events $\{\mc B_j\}_{j\geq i}$ is decreasing path by path and, hence, that the sequence $\{\ind_{\mc B_j}\}_{j\geq i}$ is nonincreasing.

Since $\mc A_j$ is measurable with respect to $m_1^{j-1}$, equations~\eqref{eq-denominator-bound} and~\eqref{eq-big-brackets} imply that for $j\geq i+1$
\[\mb E_{i+1}[\ind_{\mc A_j} \beta_j] \leq \mb E_{i+1}\left[\mb E_j\left[\ind_{\mc A_j} \frac{2C(F^k_j - F^k_{j+1})\ind_{\mc B_j}}{\hat F^k} \right]\right].\]
The law of iterated expectations then implies that
\be\label{eq-way-out} \mb E_{i+1}[\ind_{\mc A_j} \beta_j] \leq \mb E_{i+1}\left[\ind_{\mc A_j} \frac{2C(F^k_j - F^k_{j+1})\ind_{\mc B_j}}{\hat F^k}\right] = \frac{2C}{\hat F^k} \mb E_{i+1}[\ind_{\mc A_j}(F^k_j - F^k_{j+1})\ind_{\mc B_j}].\ee
Summing~\eqref{eq-way-out} over all $j\geq i+1$, we obtain
\be\label{eq-stopping}
\mb E_{i+1}\left[\sum_{j\geq i+1}\ind_{\mc A_j} \beta_j\right] \leq \frac{2C}{\hat F^k} \mb E_{i+1}\left[\sum_{j\geq i+1} \ind_{\mc A_j} \ind_{\mc B_j} (F^k_j - F^k_{j+1}(m_j))\right].\ee
Since the indicator functions on the right-hand side are nonincreasing in $j$ for $j\geq i+1$, path by path, there must be a first (possibly infinite) round $J$ for which the product of indicators is zero:
\[J = \inf\{j\geq i+1 : \ind_{\mc A_j} \ind_{\mc B_j} = 0 \}\]
with the convention that $J = +\infty$ if the set is empty.
In words, $J$ is either the first round in which $F^{k+1}_j$ exceeds $G\ve$, or the {\em last} round at which the continuation equilibrium is informative (which implies that $\gamma_j\geq 0$ for all $j> J$), whichever occurs first.\footnote{$J$ is not a stopping time with respect to the filtration $\{\mc F_j\}_{j\geq i+1}$ generated by public histories $\{m_1^{j-1}\}_{j\geq i+1}$: at the beginning of any round $j$ at which $j$ works with positive probability, it is unknown whether $j$ will be the last round in which the agent works. The proofs below do not use the optional sampling theorem or the strong Markov property.}

Consider first the paths for which $J$ is finite. The argument of $\mb E_{i+1}[\cdot]$ on the right-hand side of~\eqref{eq-stopping} then reduces to \[\sum_{j=i+1}^{J-1} (F^k_j - F^k_{j+1}) = F^k_{i+1} - F^k_{J}.\]

Consider now paths for which $J$ is infinite. In this case, the argument of $E_{i+1}[\cdot]$ on the right-hand side of~\eqref{eq-stopping} is equal to
\[\sum_{j=i+1}^\infty (F^k_j - F^k_{j+1}) = \lim_{J\rightarrow\infty}\sum_{j=i+1}^J (F^k_j - F^k_{j+1}) = \lim_{J\rightarrow\infty}\{F^k_{i+1} - F^k_{J+1}\} = F^k_{i+1} - \lim_{J\rightarrow\infty} F^k_J.\]
Notice that the limit is well defined because $F^k_J$ is a nonnegative supermartingale.\footnote{The argument is similar to the one used to prove Lemma~\ref{lem-supermart}. For any fixed $j$, let $\bar F^{k+1}_l$ denote the probability assigned by $l\geq j$ to there remaining at least $k+1$ signals {\em at the beginning of round $j$}. The process $\{\bar F^{k+1}_l\}_{l\geq j}$  is a martingale in $j$, by the law of iterated expectations and the fact that that $l$'s filtration grows finer as $l$ increases. Moreover, $\bar F^{k+1}_l \geq F^{k+1}_l$ for all $l\geq j$, because the actual number of remaining signals only decreases over time. Therefore, we have $F^{k+1}_j = \bar F^{k+1}_j = E_j[\bar F^{k+1}_{j+1}]\geq E_j[F^{k+1}_{j+1}]$. This, together with the fact that $F^{k+1}_l$ is uniformly bounded and measurable with respect to the information at the beginning round $l$, shows that it is a supermartingale.} For notational consistency, we will call this limit $F^k_{J}$, regardless of whether $J$ is finite.

Combining these observations with~\eqref{eq-stopping}, we conclude that
\[\mb E_{i+1}\left[\sum_{j>i}\ind_{\mc A_j} \beta_j\right] \leq \frac{2C}{\hat F^k} \left(F^k_{i+1} - \mb E_{i+1}[F^k_{J}]\right).\]
We have $\beta_j = \gamma_j F^1_j \lambda$. Moreover, Lemma~\ref{lem-start} shows that $\gamma_j>0$ only if $F^1_j \geq c/R$. This implies\footnote{The inequality clearly holds if $\gamma_j =0$.} that $\gamma_j \leq g \beta_j$ where $g = \frac{R}{\lambda c}$ and hence that
\be\label{eq-mid}\mb E_{i+1}\left[\sum_{j\geq i+1}\ind_{\mc A_j} \gamma_j\right] \leq \frac{2Cg}{\hat F^k}  \left(F^k_{i+1} - \mb E_{i+1}[F^k_{J}]\right).\ee

Let $\mc Z$ denote the event that at least one agent $j>i$ works and $\pi_{i+1}(m_i) = Pr_{i+1}(\mc A\cap \mc Z)$, i.e., the probability that $\mc A$ and $\mc Z$ both occur conditional on $m_1^i$. We have
\begin{align*}
\pi_{i+1}(m_i) & = Pr_{i+1}\left(\ind_{\mc A} \sum_{j>i} \ind_{\mbox{$j$ works}} \geq1\right)
\\ &  \leq \mb E_{i+1}\left[\ind_{\mc A} \sum_{j>i} \ind_{\mbox{$j$ works}} \right]
\\ & = \sum_{j>i} \mb E_{i+1}\left[\ind_{\mc A}\ind_{\mbox{$j$ works}} \right]
\\ & \leq \sum_{j>i} \mb E_{i+1}\left[\ind_{\mc A_j}\ind_{\mbox{$j$ works}} \right]
\\ & =\sum_{j>i} \mb E_{i+1} \left[\mb E_j\left[\ind_{\mc A_j}\ind_{\mbox{$j$ works}} \right]\right]
\\ & =\sum_{j>i} \mb E_{i+1} \left[\ind_{\mc A_j} \mb E_j\left[\ind_{\mbox{$j$ works}} \right]\right]
\\ & = \sum_{j>i} \mb E_{i+1}\left[\ind_{\mc A_j} \gamma_j\right].
\end{align*}
The first equality comes from the fact that $\mc Z$ is identical to the event $\{\sum_{j>i} \ind_{\mbox{$j$ works}} \geq 1\}$. The first inequality comes from the fact that $\ind_{\mc A} \sum_{j>i} \ind_{\mbox{$j$ works}}$ is nonnegative and integer valued, so that its expectation exceeds the probability that it is strictly positive. The second equality is an application of Tonelli's theorem. The second inequality comes from the fact that $\mc A\subset \mc A_j$ and, hence, $\ind_{\mc A}\leq \ind_{\mc A_j}$. The next equality comes from the law of iterated expectations and the next one comes from the fact that $\mc A_j$ is measurable with respect to the information at the beginning of round $j$. The last equality comes from the definition of $\gamma_j$.

From~\eqref{eq-mid} and Tonelli's theorem, this implies that
\be\label{eq-really-final} \pi_{i+1}(m_i)\leq \frac{2Cg}{\hat F^k} \left(F^k_{i+1} - \mb E_{i+1} [F^k_{J}]\right).\ee

Let $F^{k,r}_i(m_i)$ denote the probability that there are at least $k$ signals left conditional on $i$ working and reporting $m_i$, and for any signal $s_i$ let $F^{k,w}_i(s_i)$ denote the probability that there are at least $k$ signals left conditional on $i$ working and discovering  $s_i$. $F^{k,w}_i(s_i)$ represents $i$'s belief after discovering $s_i$, while $F^{k,r}_i(m_i)$ represents what $i+1$ would believe conditional on knowing that $i$ worked and observing report $m_i$.

Let $N_i = \{m_i: F^{k,r}_i(m_i)> (1+\eta) F^k_i\}$.

\begin{lemma} i) $\gamma_i(N_i)\leq \frac{1}{2\eta G^2}$, ii) for $m_i\notin N_i$, $\pi_{i+1}(m_i)\leq \frac{2Cg}{\hat F^k} \left((1+\eta)F^k_i - \mb E_{i+1}[F^k_{J}]\right)$.\label{lem-N-mixed}\end{lemma}

\begin{proof} i)~Let $S'_i = \{s_i : F^{k,w}_i(s_i) > F^k_i\}$, and, for each $s_i$, let $\gamma'_i(s_i)$ (resp. $\beta'_i(s_i)$) denote the probability that $i$ discovers $s_i$ given that he works (resp., the probability that $i$ works and discovers $s_i$). Also let $\gamma'_i(s_i|k+1)$ (resp. $\beta'_i(s_i|k+1)$) denote the same probabilities conditional on $|S_i|\geq k+1$.

We have for all $s_i$
\[F^{k,w}_i(s_i)  = \frac{F^{k+1}_i \gamma'_i(s_i|k+1)}{\gamma'_i(s_i)} = \frac{F^{k+1}_i \beta'_i(s_i|k+1)}{\beta'_i(s_i)}\]
where the second equality comes from $\beta'_i(s_i) = \gamma_i \gamma'_i(s_i)$ and $\beta'_i(s_i|k+1) = \gamma_i \gamma'_i(s_i|k+1)$. Therefore, $F^{k,w}_i(s_i)> F^k_i$ only if $\beta'_i(s_i) < \beta'_i(s_i|k+1)F^{k+1}_i/F^k_i$.

We have $\sum_{s_i\in S'_i} F^{k+1}_i \beta'_i(s_i|k+1)\leq \gamma_i F^{k+1}_i$. Therefore, the probability $\beta'_i(S'_i)$ that $i$ works and finds a signal in $S'_i$ satisfies
\[\beta'_i(S'_i) = \sum_{s_i\in S'_i} \beta'_i(s_i) \leq \gamma_i \frac{F^{k+1}_i}{F^k_i}.\]
Since $\beta'_i(S'_i) = \gamma_i \gamma'_i(S'_i)$, we have
\[\gamma'_i(S'_i)\leq \frac{F^{k+1}_i}{F^k_i} \leq \frac{\ve}{\hat F^k},\]
since $F^{k+1}_i\leq \ve$ and $F^k_i\geq \hat F^k$. From~\eqref{eq-bound-ve-mixed}, the right-hand side is bounded above by~$\frac{\hat F^k}{2G^2}$.

For any $m_i$, let $q(m_i)$ denote the probability, conditional on $i$ working and sending report $m_i$, that $i$ has discovered a signal $s_i\in S'_i$, and let $\sigma(s_i|m_i)$ denote the probability that $i$ discovered $s_i$ given that he worked and reported $m_i$. We also let $s_i = \ep$ denote the event that $i$ did not find anything, $\sigma(\ep|m_i)$ denote the probability that $i$ found nothing given that he worked and reported $m_i$, $F_i^{k,w}(\ep)$ denote the probability that there at least $k$ signals conditional on $i$ working and finding nothing. We have
\begin{align*}
F_i^{k,r}(m_i) & = \sum_{s_i\in S_i} \sigma(s_i|m_i) F_i^{k,w}(s_i) \\
& = \sum_{s_i \in S'_i} \sigma(s_i|m_i) F_i^{k,w}(s_i) + \sigma(\ep|m_i) F_i^{k,w}(\ep) + \sum_{s_i \neq \ep, s_i \in S_i\setminus S'_i} \sigma(s_i|m_i) F_i^{k,w}(s_i)
\end{align*}
By construction, $F_i^{k,w}(s_i)\leq F^k_i$ for all $s_i$ in the last term. Moreover, we have $F_i^{k,w}(\ep)\leq F^k_i$, as is easily checked:\footnote{Formally, for $k\geq 1$, we have $F_i^{k,w}(\ep) = \frac{F^k_i (1-\lambda)}{(1-F^1_i) + F^1_i (1-\lambda)} = F^k_i \frac{1-\lambda}{(1-F^1_i) + F^1_i(1-\lambda)} \leq F^k_i$.} intuitively, finding nothing always increases the probability that there are no signals remaining to be found. Finally, the first term is bounded above by $\sigma(S'_i|m_i) = q(m_i)$. Therefore,
\begin{align*}
F_i^{k,r}(m_i)\geq (1+\eta) F^k_i & \Rightarrow q(m_i) + (1-q(m_i))F^{k}_i \geq (1+\eta) F^k_i\\
& \Rightarrow q(m_i)\geq \eta F^k_i.
\end{align*}
To conclude, note that
\[\sum_{m_i} \gamma_i(m_i) q(m_i) = Pr(s_i\in S'_i|m_1^{i-1}, \mbox{$i$ works}) \leq \frac{\hat F^k}{2G^2}\]
The left-hand side is bounded below by $\gamma(N_i) \eta F^k_i$.
Since $F^k_i\geq \hat F^k$, this implies that
\[\gamma(N_i)\leq  \frac{1}{2\eta G^2}.\]

ii)~$F^k_{i+1}(m_i) $ is a convex combination\footnote{$F^k_{i+1}(m_i)$ is the probability that $i+1$ assigns to there being at least $k$ signals left upon observing $m_i$. If $i+1$ knew that $i$ didn't work and simply sent message $m_i$, this belief should be $F^k_i$ since $m_i$ conveys no additional information. And if $i+1$ knew that $i$ produced $m_i$ through working and then reporting $m_i$, his updated belief should be $F^{k,r}_i(m_i)$. Since $i+1$ doesn't observe $i$'s action, in general $F^k_{i+1}(m_i)$ is a convex combination of these two posteriors, where the weights corresponds to the probability assigned by $i+1$ to $i$ fabricating or working conditional on observing $m_i$. This fact is straightforward to check using Bayesian updating.} of $F^k_i$ and $F^{k,r}_i(m_i)$, we have $F^k_{i+1}(m_i)\leq (1 + \eta) F^k_i$ for all $m_i\notin N_i$. From~\eqref{eq-really-final}, this implies that
\[\pi_{i+1}(m_i) \leq \frac{2Cg}{\hat F^k} \left((1+\eta) F^k_i - \mb E_{i+1}[F^k_{J}]\right)\]
for all $m_i\notin N_i$.\hfill$\blacksquare$
\end{proof}

Let $T_i = \{m_i: F^{k+1}_{i+1}(m_i) \geq \sqrt{G} \ve\}$.

\begin{lemma} i) $Pr_{i+1}(\mc A)\geq 1-1/\sqrt{G}$ for all $m_i\notin T_i$, ii) $\gamma_i(T_i)\leq 1/\sqrt{G}$.\label{lem-T-mixed}\end{lemma}
\begin{proof}
 i) If $m_i\notin T_i$, we have $F^{k+1}_{i+1} \leq \sqrt{G}\ve$. Using this inequality in Lemma~\ref{lem-supermart} instead of $F^{k+1}_i\leq \ve$ and repeating its argument applied to round $i+1$, we conclude that $i+1$ assigns probability at least $1- 1/\sqrt{G}$ to $\mc A$ whenever $m_i\notin T_i$.

ii)~Let $F^{k+1,r}_i(m_i)$ denote the probability that there are at least $k+1$ signals left at the beginning of round $i+1$ conditional on $i$ working and reporting $m_i$. $F^{k+1}_{i+1}$ is a convex combination of $F^{k+1}_i$ and $F^{k+1,r}_i(m_i)$. This, together with the fact that $F^{k+1}_i \leq \ve$ and the definition of $T_i$, shows that $m_i\in T_i$ only if $F^{k+1,r}_i(m_i)\geq \sqrt{G} \ve$. Let $T'_i$ denote the set of messages $m_i$ for which the last inequality holds. As noted, $T_i\subset T'_i$.

Since $i$'s prior probability that there are at least $k+1$ signals left is $F^{k+1}_i\leq \ve$, the law of iterated expectations implies that the probability $\bar F^{k+1,r}_i(m_i)$ that there were at least $k+1$ signals left at the beginning of round $i$
conditional on $i$ working and finding $m_i$ must satisfy
\[\mb E_i[\bar F^{k+1,r}|\mbox{working}] = \sum_{i\in M_i} \gamma_i(m_i) \bar F^{k+1,r}_i(m_i) = F^{k+1}_i \leq \ve.\]
Using Markov's inequality, this implies that $Pr(m_i : \bar F^{k+1,r}_i(m_i) \geq \sqrt{G}\ve | \mbox{ $i$ works})\leq \frac{\ve}{\sqrt{G}\ve} = 1/\sqrt{G}.$ Since also $\bar F^{k+1,r}_i(m_i)\geq F^{k+1,r}_i(m_i)$, we get $\gamma_i(T'_i)\leq 1/\sqrt{G}$. Since $T_i\subset T'_i$, this shows that $\gamma_i(T_i)\leq 1/\sqrt{G}$.\hfill$\blacksquare$\end{proof}

Let $V^w$ denote $i$'s expected gross utility if he works and $V^w(m_i)$ denote his expected gross utility conditional on working and reporting $m_i$. We have
\begin{align}
V^w  & = \sum_{i\in M_i} \gamma_i(m_i) V^w(m_i) \nonumber \\
& \leq  (\gamma_i(N_i) + \gamma_i(T_i))R + \sum_{m_i\notin N_i\cup T_i} \gamma_i(m_i) V^w(m_i) \nonumber \\
& \leq \left(\frac{1}{2\eta G^2} + \frac{1}{\sqrt{G}} \right)R + \sum_{m_i\notin N_i\cup T_i} \gamma_i(m_i) V^w(m_i) \label{eq-total-comp-mixed}
\end{align}

Moreover,
\be\label{eq-comp}V^w(m_i) = Pr(\mc Z|m_1^i) V^w(m_i|\mc Z) + Pr(\mc Z^c|m_1^i) V^w(m_i|\mc Z^c).\ee
Conditional on $m_1^i$, the event $\mc Z$ is independent of how $m_i$ was produced (i.e., whether $m_i$ was obtained by work or fabrication). Indeed, as long as no one works, the distribution of reports $m_j$ made by agents following $i$ depends only on $m_1^i$, not on the signals that remain to be discovered in the case. And as soon as someone works, then by definition $\mc Z$ has occurred. Thus, what triggers the event $\mc Z$ (whenever it occurs) is a sequence of uninformative (until $\mc Z$ occurs) reports $m_j$ for agents following $i$, whose probability distribution is completely pinned down by $m_1^i$.

From the previous lemmas, we have $Pr(\mc A\cap \mc Z|m_1^i)\leq \frac{2Cg}{\hat F^k} \left((1+\eta) F^k_i - \mb E_{i+1}[F^k_{J}]\right)$ for all $m_i\notin N_i$ and $Pr(\mc A^c|m_1^i)\leq 1/\sqrt{G}$ for all $m_i\notin T_i$. Letting $\hat M_i = M_i\setminus (N_i\cup T_i)$, this implies that
\be\label{eq-pro-Z-mixed} Pr(\mc Z|m_1^i) = Pr(\mc Z \cap \mc A |m_1^i) +  Pr(\mc Z \cap \mc A^c|m_1^i)\leq \frac{2Cg}{\hat F^k} \left(F^k_i(1+\eta) - \mb E_{i+1}[F^k_{J}]\right) + 1/\sqrt{G}\ee
for all $m_i\in \hat M_i$.

Conditional on $\mc A$, $F^{k+1}_{j}\leq G\ve$ for all $j\geq i$. By definition of $J$, all continuation equilibria until round $J$ included are informative, which implies that $F^k_j\geq F^k(F^{k+1}_j)$ for all $j\leq J$. Since $F^k(\cdot)$ is nonincreasing, this implies that $F^k_j\geq F^k(G\ve) = \hat F^k$ for all $j\leq J$.

We thus have for $m_i\in \hat M_i$
\begin{align*}
\mb E_{i+1} F^k_{J} & = Pr_{i+1}(\mc A) \mb E_{i+1}[F^k_{J}|\mc A] + Pr_{i+1}(\mc A^c) \mb E_{i+1}[F^k_{J}|\mc A^c]\\
& \geq Pr_{i+1}(\mc A) \mb E_{i+1}[F^k_{J}|\mc A]\\
& \geq (1-1/\sqrt G) \hat F^k.
\end{align*}

By construction, $F^k_i\leq \bar F^k +\hat F^k \eta$ and $\hat F^k\geq \bar F^k - \hat F^k\eta$ and hence $F^k_i -\hat F^k \leq (\bar F^k + \hat F^k \eta) - (\bar F^k - \hat F^k \eta) = 2\eta \hat  F^k$. Letting $B = 8Cg$,~\eqref{eq-pro-Z-mixed} then implies (for $\eta\leq 1/2$, which we assume) that for all $m_i$ in $\hat M_i$
\be\label{eq-MOVE} Pr(\mc Z|m_1^i)\leq B \eta  + \frac{B}{2\sqrt G} + 1/\sqrt{G}.\ee

For each $m_i$, let $V^{f}_i(m_i|\mc Z^c)$ denote $i$'s expected gross utility if he sends message $m_i$ conditional on no $j>i$ working and $V^{f,*}$ denote the maximizer of $V^f_i(m_i|\mc Z^c)$ over all messages $m_i\in \hat M_i$. Notice that $i$'s expected gross utility conditional on $m_i$ and no $j>i$ working does not depend on whether $i$ worked or shirked: either way, the subsequent reports $\{m_j\}_{j>i}$ are independent of the signals that remain to be discovered. Therefore, $i$'s conditional expected gross utilities satisfy $V^w(m_i|\mc Z^c) = V^f_i(m_i|\mc Z^c)$.

Combining these observations with~\eqref{eq-comp} and~\eqref{eq-MOVE}, we obtain
\[V^w(m_i) \leq \left(B\eta + \frac{B}{2\sqrt G} +\frac{1}{\sqrt{G}}\right) R + V^{f,*},\]
for $m_i\in \hat M_i$. Combining this with~\eqref{eq-total-comp-mixed} yields
\[V^w\leq \left(\frac{1}{2\eta G^2} + \frac{1}{\sqrt{G}}\right) R+ \left(B\eta+\frac{B}{2\sqrt G} + \frac{1}{\sqrt{G}}\right)R + V^{f,*}.\]
$i$'s utility from working thus satisfies
\be\label{eq-i-utility-end-mixed}U^w\leq \left(\frac{1}{2 \eta G^2} + \frac{1}{\sqrt{G}}+ 2B\eta+\frac{B}{2\sqrt G}+ \frac{1}{\sqrt{G}}\right)R + V^{f,*} - c.\ee

If $i$ sends a message $m_i^*\in \hat M_i$ that achieves $V^{f,*}$, his utility $U^f$ satisfies
\begin{align*}
U^f & \geq  Pr(\mc Z|m_1^{i-1}, m_i^*) \times 0 + Pr(\mc Z^c|m_1^{i-1}, m_i^*) V^{f,*}\\
& \geq \left(1-B\eta - \frac{B}{2\sqrt G }- \frac{1}{\sqrt{G}}\right) V^{f,*}\\
& \geq V^{f,*} - \left(B\eta- \frac{B}{2\sqrt G}- \frac{1}{\sqrt{G}}\right) R
\end{align*}
where 0 is used as a lower bound on $i$'s realized gross utility to derive the second inequality.

Therefore, working is strictly suboptimal if
\[\left(\frac{1}{2\eta G^2} + \frac{1}{\sqrt{G}} + B\eta+\frac{B}{2\sqrt G}+ \frac{1}{\sqrt{G}}\right) R + V^{f,*} - c < V^{f*} - (B\eta +\frac{B}{2\sqrt G}+ \frac{1}{\sqrt{G}})R\]
or
\be\label{eq-final-bound}\frac{c}{R}> \frac{1}{2 \eta G^2} + \frac{3}{\sqrt{G}} + 2 B\eta + \frac{B}{\sqrt G}.\ee
This inequality is always satisfied as long as $\eta$ is small enough and $\eta G$ is large enough. In particular, since $B = 8Cg$, $g = R/\lambda c$,  and $C$ can be taken to equal $2R/c$ if $\lambda =1$ and $2R/(c\lambda(1-\lambda))$ as noted in Proposition~\ref{pro-general-mixed}, the inequality is satisfied if $\eta = 1/\sqrt{G}$ and $\sqrt{G} = 128 R^3/ c^3$ for  $\lambda =1$ and $\sqrt{G} = 128 R^3/(c^3 \lambda^2(1-\lambda))$ if $\lambda <1$ (recalling that $R>c)$, in which case each term on the right-hand side of~\eqref{eq-final-bound} is less than $c/4R$ with some strict inequalities.\hfill$\blacksquare$

\section{Proof of Theorem~\ref{the-pos-statement}}

We construct compensation functions for which the strategy profile proposed constitutes an equilibrium. Under this strategy profile, as long as $p_i$ lies in $(\lbar p,\bar p)$ all agents work and report an informative signal about $\omega$. Moreover, given the symmetric signal structure, $p_i$ depends only on the number of ``$H$'' and ``$L$''  signals as long as all agents $j<i$ work with probability 1. Therefore, the set of equilibrium posteriors forms a grid $\{q^k\}$ containing $\hat p$ and containing a single point on each side of $(\lbar p, \bar p)$. Let $q^0\leq \lbar p < q^1, \ldots , \hat p, \ldots, q^{N} < \bar p < q^{N+1}$ denote this grid. Along the candidate equilibrium, the belief $p_i$ evolves on this grid until it hits either $q^0$ or $q^{N+1}$, after which the investigation stops.

Let $J$ denote the last investigator who works:  we have $p_{J}\in \{q^1, q^N\}$ and $p_{J+1}\in \{q^0, q^{N+1}\}$.
Also let $\tilde p = p_{J+1}$ denote the value of the belief when learning stops under the candidate equilibrium.

We construct utility functions in which an investigator's compensation depends only on his report and on the posterior $\tilde p$.

For any $i$ such that $p_i = q^k \in (\lbar p, \bar p)$, if $i$ reports ``$H$'', he gets a reward $R^k_H \geq 0$ if $\tilde p = q^{N+1}$ and a punishment $P_L^k \leq 0$ if $\tilde p = q^0$. If $i$ reports ``$L$'', he gets $R_L^k\geq 0$ if $\tilde p = q^0$ and $P_L^k \leq 0$ if $\tilde p = q^{N+1}$.

For any $p,q$ on the grid, let $\pi(p,q)$ denote the probability that the belief sequence ends with $\tilde p = q^{N+1}$, i.e., exits $(\lbar p, \bar p)$ through $\bar p$, from the perspective of an agent who assigns probability $p$ to $\omega$, but the prior used by investigators is $p_0 = q$. That is, $\pi(p,q)$ is the probability that an individual with prior $p$ assigns to the sequence $p_i$ converging to $q^{N+1}$ in equilibrium given that the public belief, which serves as the state variable for the equilibrium, starts at $q$.

If $i$ sends report ``$H$'' starting from prior $p_i=q^k$, he assigns a probability  $\pi(q^k, q^{k+1})$ to the public belief converging to $q^{N+1}$. If $i$ works and receives report $``H''$, his belief about the continuation equilibrium is $\pi(q^{k+1}, q^{k+1})$. Similarly, if $i$ sends $``L''$, his belief is $\pi(q^k, q^{k-1})$ whereas if he works and reports $``L''$ his belief is $\pi(q^{k-1}, q^{k-1})$. It is straightforward to verify the inequalities
\be\label{eq-cheat-up}\pi(q^{k+1}, q^{k+1}) > \pi(q^k, q^{k+1})\ee
and
\be\label{eq-cheat-down}\pi(q^{k-1}, q^{k-1}) < \pi(q^k, q^{k-1}),\ee
for all $k\in [2, N-1]$. The strictness of the inequalities comes from the fact that conditional on the true state $\omega$, the dynamic of $\{p_j\}_{j\geq i+1}$ starting any given value of $p_{i+1}$ is strictly increasing in $\omega$ in FOSD, as is easily checked. Therefore, the probability of hitting $q^{N+1}$ before $q^0$ is strictly increasing in the belief $p_i$ that the state is high.

For $k =1$, the investigation stops if $i$ reports ``$L$'' so~\eqref{eq-cheat-down} holds as an equality, but~\eqref{eq-cheat-up} is still strict, because this report triggers further investigation. The reverse is true for $k=N$:~\eqref{eq-cheat-up} only holds as an equality while~\eqref{eq-cheat-down} is strict.

If $i$ shirks, his maximal utility is
\be\label{eq-payoff-shirk}\max\{\pi(q^k, q^{k+1}) R^k_H + (1-\pi(q^k, q^{k+1}))P^k_H; \pi(q^k, q^{k-1}) P^k_L + (1-\pi(q^k, q^{k-1}))R^k_L \}.\ee
The left argument is $i$'s expected payoff if he sends ``$H$'', and the right one is his payoff if he sends ``$L$''. Since $i$ can send either message at no cost, his best payoff from fabrication is the maximum of these two terms.
If $i$ works, he gets
\be\label{eq-payoff-work} z^k [\pi(q^{k+1}, q^{k+1}) R^k_H + (1-\pi(q^k, q^{k+1}))P^k_H] + (1-z^k) [\pi(q^{k-1}, q^{k-1}) P^k_L + (1-\pi(q^{k-1}, ^{k-1}))R^k_L]\ee
where $z^k$ is the probability of receiving signal ``$H$'' given belief $q^k$, and is equal to $z^k = Pr(``H''|q^k) = q^k \pi + (1-q^k) \times (1-\pi) $.

It is optimal for $i$ to work if the expression in~\eqref{eq-payoff-work} exceeds the expression in~\eqref{eq-payoff-shirk} by at least $c$.

This condition is obtained as follows: set $P^k_H= P^k_L= - Q$ where $Q$ is a strictly positive constant, and let $R^k_H = Q \frac{1-\pi(q^k, q^{k+1})}{\pi(q^k, q^{k+1})}$ and $R^k_L = Q\frac{\pi(q^k, q^{k-1})}{1-\pi(q^k, q^{k-1})}$. This guarantees that $i$'s expected payoff from fabrication is zero, regardless of the outcome. From~\eqref{eq-cheat-up} and~\eqref{eq-cheat-down}, his payoff from working is of order  $Q$ and thus exceeds $c$, for $Q$ high enough.\footnote{To see this, let $\bar \pi$ denote a strictly positive lower bound on all inequalities~\eqref{eq-cheat-up} and~\eqref{eq-cheat-down} over all $k$'s whenever they hold strictly. Then, the gain from working is of order $Q \bar \pi$.}

If $k=1$ or $N$, there is one signal that $i$ can send after working which yields a payoff of order $Q$, while the other signal yields 0. The signal associated with a positive payoff arises with a probability that is bounded away from 0, since $p_i$ lies in $(\lbar p, \bar p)$.

Moreover this scheme is feasible as long as the maximal reward $R$ and and punishment $P$ respectively exceed $\sup\{R^k_\theta: \theta\in \{L,H\}, k\in \{1,\ldots, N\}\}$ and $Q$.

\subsection*{General case}

As noted, the argument extends easily to the case in which $\rho$ and/or $\lambda$ are less than 1.

With $\rho < 1$ and $\lambda=1$, the informative equilibrium is identical to the one described in Proposition~\ref{pro-positive} except that learning stops as soon as an agent fails to report evidence, in which case he gets a zero compensation. By construction of the equilibrium in the proof above, shirking and reporting that no evidence was found has the same value as fabricating any other message and can thus be deterred. Since a working agent may find nothing, or the learning process may be interrupted before the belief process exits $(\lbar p,\bar p)$, in which case the working agent receives 0, the rewards and punishments must be scaled up by $1/\pi_\rho(q^k)$, where $\pi_\rho(q^k)$ is the probability that the belief process exits $(\lbar p, \bar p)$ in equilibrium, given the current belief $q^k$, so that the expected compensation of a working agent still exceeds the cost $c$ of working.

If $\lambda <1$ and $\rho=1$, a working agent may fail to find evidence even when there surely exists some. In this case, we again assume that the compensation is zero, which deters shirking and reporting the empty message, and scale up all rewards and punishments by $1/\lambda$ to incentive the agent to work, as in the previous paragraph. The belief process will surely exit $(\lbar p, \bar p)$ since the amount of evidence is unlimited (only individual agents may be unlucky and find nothing with probability $1-\lambda$).

The case in which both $\lambda$ and $\rho$ are less than 1 is a convex combination of the previous cases and addressed accordingly.

\section{Proof of Theorem~\ref{the-witness}}

Without loss of generality, we again assume that agents' gross utility functions $\{V_i\}_{i\in \mb N}$ all take values in $[0,R]$.

For any round $i$, consider the event $\mc Q_i$ that all past investigators whose may have failed to discover signals, given the history $m_1^{i-1}$, have indeed failed to discover signals. Conditional on $\mc Q_i$, the number $q_i$ of signals that have been uncovered until round $i$ is equal to the number of past witnesses plus the number of past investigators whose messages reveal that they have surely discovered signals given $m_1^{i-1}$. Put differently, $q_i$ is the number of signals that have been surely discovered by round $i$. Let $\hat F^k_i$ denote the probability that $|S_i|\geq k$ conditional on $\mc Q_i$ and $q_i$.

We observe that i) $q_i$ is nondecreasing along any equilibrium path and is strictly increasing whenever a witness arrives, ii) given the construction of $S$, the distribution of $S_i$ conditional on $m_1^{i-1}$ and $\mc Q_i$ is only a function of $q_i$, iii) $\hat F^k_1 = F^k_1$ for all $k$.

We will prove that there exist strictly positive thresholds $\{\lbar F^k\}_{k\geq 1}$ such that an informative continuation equilibrium exists in round $i$ only if $\hat F^k_i\geq \lbar F^k$ for all $k\geq 1$. Applied to $i=1$, this result implies Theorem~\ref{the-witness}. The proof uses the following lemma.

\begin{lemma}\label{lem-IHR}  Under Assumption~\ref{A-IHR},  the following inequalities hold: i) $F^k_i\leq \hat F^k_i$ for all $i,k\geq 1$, and ii) $\hat F^k_j\leq \hat F^k_i$ for all $j\geq i$ and $k\geq 1$.\end{lemma}
\begin{proof} Part~i) Let $r_i$ denote the number of signals discovered by round $i$. We have $r_i \geq q_i$ and
\begin{align*} F^k_i & = \sum_{r \in \{q_i,\ldots, i-1\}} Pr(r_i = r|m_1^{i-1}) Pr(\tilde K\geq r + k | \quad \tilde K\geq r)\\
& \leq \sum_{r \in \{q_i,\ldots, i-1\}} Pr(r_i = r|m_1^{i-1})  Pr(\tilde K\geq q_i+k | \quad \tilde K\geq q_i)\\
& = \sum_{r \in \{q_i,\ldots, i-1\}} Pr(r_i = r|m_1^{i-1}) \hat F^k_i \\
& = \hat F^k_i.
\end{align*}
The first equality comes from the independence of $\tilde K$ from $S^\infty$: $r_i$ is a sufficient statistic for $\tilde K$ given all the information produced before round $i$, and only to the extent that it reveals that $\tilde K\geq r_i$. The inequality comes from the increasing hazard rate condition, which implies that for any $k\geq 0$, $Pr(\tilde K\geq k+q \enskip | \enskip \tilde K\geq q)$ is non-increasing in $q$.\footnote{\label{ftn-IHR}See, e.g., Barlow  et al.  (1963, p. 379). In brief, a random variable $X$ with distribution $F$ has the increasing hazard rate property if and only if the survival distribution $\bar F = 1-F$ is log-concave. This property implies, as is easily checked, that for any $p,q\geq 0$, $Pr(X\geq p+q| X\geq p) = \bar F(p+q)/\bar F(p)$ is decreasing in $p$.} The second equality is due to the equality $\hat F^k_i = Pr(\tilde K\geq k+q_i \enskip | \enskip \tilde K\geq q_i)$, which again comes from the independence of $\tilde K$ and $S^\infty$: the only relevant information about $\tilde K$ conditional on $\mc Q_i$ is the number of signals $q_i$ discovered by round $i$.

Part ii) For any $j\geq i$, we have $q_j\geq q_i$. As explained in the proof of Part i), we have for all $k\geq 1$
\begin{align*} \hat F^k_j &= Pr(\tilde K\geq k+q_j| \tilde K \geq q_j)\\ &\leq Pr(\tilde K\geq k+q_i|\tilde K\geq q_i)\\ & = \hat F^k_i,\end{align*}
where the inequality comes from the increasing hazard rate property (see Footnote~\ref{ftn-IHR}).\hfill$\blacksquare$
\end{proof}

The existence of thresholds $\{\lbar F^k\}_{k\geq 1}$ in the result mentioned above is proved by induction on $k$. We start with the base case $k=1$ and then prove the induction step.

\subsection{Proof for $k=1$}

Suppose that $\hat F^1_i < c/2R$. Combining the two parts of Lemma~\ref{lem-IHR}, this implies that $F^1_j< c/2R$ for all $j\geq i$. Lemma~\ref{lem-start} still applies: no investigator $j\geq i$ works because the probability that he finds something is too small to justify the cost of effort, given the maximal reward $R$.

We now show that if $\hat F^1_i$ lies below another threshold, smaller than $c/2R$, witnesses provide no informative message, either.

If $i$ is a witness, we will use the following notation:
\bit
\vspace{-.2cm}
\item $\beta_i$: probability that $i$ produces an informative message given $m_1^{i-1}$;
\item $M_i^+$: set of messages $m_i$ which are followed by an informative continuation equilibrium (as in the case of investigators);
\item $Pr_i(M_i+)$: probability that $i$ produces a message in $M_i^+$ given $m_1^{i-1}$;
\item $\gamma_i(m_i)$: probability that $i$ sends $m_i$ given $m_1^{i-1}$;
\item $\gamma_i(m_i|s_i)$: probability that $i$ sends $m_i$ after observing $s_i$
\eit

\begin{lemma}\label{lem-order-stat} Consider $L\geq 2$ pairwise independent random variables $\{Y_\ell\}$ with non-atomic distributions over~$\mb R$ and densities $f_\ell$ that are bounded above by $\bar f$.  For any $\ve\geq 0$, let $E^L_\ve$ denote the event that $\exists \ell,\ell'\leq L$ such that $|Y_\ell-Y_{\ell'}|\leq \ve$. Then
\[Pr(E^L_\ve) \leq L(L-1) \bar f \ve.\]
\end{lemma}
\begin{proof} The result is proved by induction on $L\geq 2$. For $L=2$, we have
\[Pr(|Y-Y'|\leq \ve) =\int_{\mb R} f_Y(x) F_{Y'}[x-\ve, x+\ve] dx \leq 2\ve \bar f \int_{\mb R}  f_Y(x) dx = 2\ve \bar f.\]
Now suppose that the claim holds for $L-1$. Notice that the event $E^L_\ve$ is the union of $L$ events: the event $E^{L-1}_\ve$ concerning the first $L-1$ random variables, and, for each $\ell\leq L-1$, the event $E^{\ell,L}$ that the $L^{th}$ random variable lies within $\ve$ of the $\ell^{th}$ random variable. Therefore,
\begin{align*} Pr(E^L_\ve) & \leq Pr(E^{L-1}_\ve) + \sum_{\ell\leq L-1} Pr(|Y_\ell - Y_L|\leq \ve) \\
& \leq (L-1)(L-2)\bar f \ve  + (L-1) \times  2\ve \bar f\\
& = L(L-1) \bar f \ve,\end{align*}
where the second inequality comes from the induction hypothesis and the fact that $Pr(|Y_\ell -Y_L|\leq \ve) \leq 2 \bar f \ve$, as shown in the first step of the induction for $L=2$.\hfill$\blacksquare$
\end{proof}

\begin{lemma}\label{lem-start-wit} There exists a threshold $\lbar F^1\in (0,  c/2R)$ such that $\beta_i>0$ in some witness round $i$ only if $\hat F^1_i\geq \lbar F^1$.\end{lemma}

\begin{proof} Recall that if $i$ is a witness, his message $m_i$ is informative (in equilibrium) if it is statistically dependent of $S$ conditional on $m_1^{i-1}$.

If $i$ is a witness and has preference $\epsilon$, his message is said to be $\epsilon$-informative if there exist two signals $s_i\neq s'_i$ such that the equilibrium distributions of $m_i$ conditional on $i$ getting signals $s_i$ and $s'_i$  given $i$'s preference $\epsilon$, are different across the two signals.

The following observation is straightforward to prove.
\begin{observation} $i$'s message is informative if and only if the set of preference shocks $\epsilon$ for which $i$'s message is $\epsilon$-informative has positive probability.\end{observation}

For any equilibrium and round $i$ in which $i$ is a witness, let $\nu_i$ denote the probability that $i$'s preference shock $\epsilon_i$ is such that $i$'s message is $\epsilon_i$-informative.

For any $F< c/2R$, let $\nu(F)$ denote the supremum of $\nu_i$ over all witness rounds $i$ of all equilibria such that $\hat F^1_i\leq F$. We will show that $\nu(F)= 0$ for all $F$ below some strictly positive threshold.

Consider such an equilibrium. For any witness round $i$ and message $m_i$, let $z(m_i)$ denote the probability that at least some $j>i$ produces an informative message following message $m_i$. %Note that $p(m_i)>0$ only if $m_i\in M_i^+$.

$i$'s expected utility if he receives signal $s_i$ and sends message $m_i$ is given by
\be\label{eq-Ui-wit}U_i(m_i; s_i) = z(m_i) \mb E_i[V_i(m)|s_i, m_1^i] + (1-z(m_i))\lbar V_i(m_i) + \epsilon_i(m_i)\ee
where
\[\lbar V_i(m_i) = \mb E_i[V_i(m)|m_1^i,\mbox{no $j>i$ produces an informative message}].\]
Notice that $\lbar V_i(m_i)$ does not depend on the signal $s_i$ since this signal is payoff irrelevant whenever no $j>i$ produces an informative message.

Given $\epsilon_i$, $i$ sends an informative signal only if there exist $m_i\neq m'_i$ and signals $s_i\neq s'_i$ such that
\[U_i(m_i;s_i)\geq U_i(m'_i; s_i)\]
and
\[U_i(m'_i;s'_i)\geq U_i(m_i; s'_i)\]
From~\eqref{eq-Ui-wit} and the fact that $V_i(m)\in [0,R]$, this is possible only if
\[|\epsilon_i(m_i)+\lbar V_i(m_i) - \epsilon_i(m'_i)+\lbar V_i(m'_i)|\leq R(z(m_i) + z(m'_i))\]
The random variables $Y_\ell = \epsilon_i(m_\ell)+\lbar V_i(m_\ell)$ satisfy the assumptions of Lemma~\ref{lem-order-stat}. Letting $|M|$ denote the cardinality of the message space $M$, we thus have
\be\label{eq-info-message}Pr(\mbox{$i$ sends an informative message}) \leq |M|^2 \bar f R (z(m_i) + z(m'_i))\ee

Since no investigator $j\geq i$ works when $\hat F^1_i<c/2R$, we have for any message $m_i$
\be\label{eq-z-wit} z(m_i)\leq \sum_{j\geq 1} Pr(\mbox{there are $j$ witnesses in the sequence after round $i$, given message $m_i$}) j \nu(F).\ee
Indeed, by definition of $\nu(F)$ and the fact that $\hat F^1_j\leq F$ for all $j\geq i$, a witness provides an informative signal with probability at most $\nu(F)$. Thus $j\nu(F)$ is an upper bound on the probability that at least one witness provides an informative signal given there are $j$ such witnesses.

The probability that at least $j$ witnesses come after round $i$ is bounded above by $F^j$, where $j$ is an exponent (not a superscript). To show this for $j=1$, notice that by Lemma~\ref{lem-IHR}, $\hat F^1_{i+1}\leq F$, and the probability that there is at least one witness after round $i$ is bounded above by the probability $F^1_{i+1}$ that there is at least one more signal, which is less than $\hat F^1_{i+1}$ again by Lemma~\ref{lem-IHR}. For $j=2$, note that conditional on the first witness arriving, Lemma~\ref{lem-IHR} implies that the probability that a second witness arrives is again by bounded by $F$ since the probability that there remains another signal is bounded by $F$, and a witness can arise only if such a signal exists. By induction, this shows that the probability of having at least $j$ witnesses and, hence, the probability of having exactly $j$ witnesses, are bounded above by $F^j$. Combining this with~\eqref{eq-z-wit} and using the standard formula $\sum_{j\geq 1} j x^j = x/(1-x)^2$ for all $x\in (0,1)$, we get
\[z(m_i)\leq \sum_{j\geq 1} F^j j \nu(F) =\frac{F \nu(F) }{(1-F)^2},\]

Combining this with~\eqref{eq-info-message}, we obtain
\be\label{eq-bound-info-message}Pr(\mbox{$i$ sends an informative message} \enskip| \enskip \hat F^1_i\leq F) \leq 2 |M|^2 \bar f R \nu(F) F /(1-F)^2\ee
Taking the supremum of the left-hand side over all witness rounds $i$ and equilibria such that $\hat F^1_i\leq F$, we obtain
\[\nu(F)\leq \frac{2 \nu(F)|M|^2 \bar f R F}{(1-F)^2}.\]
For $2 F/(1-F)^2\leq 1/|M|^2\bar f R$, this relation is possible only if $\nu(F) = 0$, because the function on the right-hand side is a contraction of $\nu(F)$. The function $F/(1-F)^2$ is increasing on $[0,1)$ and starts at zero. Therefore, we conclude that there exists a threshold $\lbar F^1 > 0$ such that $\nu(F)=0$ for all $F\leq \lbar F^1$.\hfill$\blacksquare$
\end{proof}

\subsection{Induction Step}

Suppose that there exist strictly positive thresholds $\{\lbar F^{k'}\}_{k'\in \{1,\ldots, k\}}$ such that a continuation equilibrium starting at round $i$ is informative only if $\hat F^{k'}_i \geq  \lbar F^{k'}$ for all $k'\leq k$. We will show that a similar condition holds for $k+1$. The proof works by contradiction: we will suppose that for all $\ve\in (0,1)$, there exists an informative continuation equilibrium such that $\hat F^{k+1}_i\leq \ve$ and obtain an impossibility for $\ve$ small enough.

Consider any $\ve < \lbar F^k \times \lbar F^1$ and any informative continuation equilibrium starting at round $i$ such that $\hat F^{k+1}_i\leq \ve$. Then, all continuation equilibria are uninformative as soon as the witness $j\geq i$ arrives. To see this, suppose that some witness arrives at round $j\geq i$. We have for any message $m_j$ sent by this witness
\begin{align*}  \hat F^k_{j+1}(m_j) & = Pr(\tilde K\geq k+q_j +1\enskip | \enskip \tilde K \geq q_j+1)\\
 & = \frac{Pr(\tilde K \geq k+ q_j+1)}{Pr(\tilde K \geq q_j+1)}\\
 & \leq \frac{Pr(\tilde K \geq k+ q_i+1)}{Pr(\tilde K \geq q_i+1)}\\
 & =  \frac{Pr(\tilde K \geq k+ q_i+1)}{Pr(\tilde K \geq q_i)} \enskip \frac{Pr(\tilde K \geq q_i)}{Pr(\tilde K \geq q_i+1)}\\
 & = \frac{\hat F^{k+1}_i}{\hat F^1_i},
\end{align*}
where the inequality comes for the monotone hazard rate assumption (see Footnote~\ref{ftn-IHR}) and the fact that $q_j\geq q_i$. Since the continuation equilibrium from round $i$ is informative, we must have $\hat F^1_i\geq \lbar F^1$. Therefore, $\hat F^{k+1}_i\leq \ve< \lbar F^k \lbar F^1$ implies that
\[\hat F^k_{j+1}< \lbar F^k\]
which shows, by the induction hypothesis, that all continuation equilibria are uninformative from round $j+1$ onwards.

Moreover, the first witness, $j$, knowing that continuation equilibria are uninformative regardless of his message, has no incentive to send an informative message.\footnote{Formally, the argument is similar to the proof of Lemma~\ref{lem-start-wit} except that here $z(m_i) = 0$ regardless of the message.}

Therefore, if $\hat F^{k+1}_i\leq \ve$, the only agents who may send informative messages are the investigators arriving between round $i$ and the arrival of the first witness. The situation is therefore almost identical to the setting of Theorem~\ref{the-main-result}, in the absence of witnesses, except that any sequential learning activity is interrupted at the apparition of the first witness. We know from Theorem~\ref{the-main-result} that such equilibria can be informative only if $F^{k+1}_i$ exceeds the $k+1$-threshold given by Theorem~\ref{the-main-result}, which we denote here $\tilde F^{k+1}$. Since $F^{k+1}_i\leq \hat F^{k+1}_i$ by Part i) of Lemma~\ref{lem-IHR}, letting $\lbar F^{k+1} =\min\{\tilde F^{k+1}_i,  \lbar F^k \lbar F^1\}$ we conclude that no informative continuation equilibrium exists in round $i$ if $\hat F^{k+1}_i\leq \lbar F^{k+1}$. This concludes the induction step.

\end{document}